\documentclass[11pt]{article}
\usepackage{natbib}
\usepackage{booktabs}
\usepackage[ruled]{algorithm2e}
\usepackage{fullpage}

\SetAlFnt{\small}
\SetAlCapFnt{\small}
\SetAlCapNameFnt{\small}
\SetAlCapHSkip{0pt}
\IncMargin{-\parindent}

\usepackage{amsmath,amsthm,amsfonts}
\usepackage{tikz}
\usetikzlibrary{calc}
\usepackage{pgfplots}
\usepackage{url}
\usepackage{lipsum}

\pgfplotsset{compat=newest}
\tikzset{
	pics/myrec/.style n args={3}{code={
			\draw[draw=none, #3] (0,0) rectangle (#1,#2);
	}},
	pics/myrec/.default={1}{0}{pink},
}

\makeatletter
\renewcommand\paragraph{
	\@startsection{paragraph}
	{4}
	{\z@}
	{3.25ex \@plus1ex \@minus.2ex}
	{-1em}
	{\normalfont\normalsize\bfseries\maybe@addperiod}
}
\newcommand{\maybe@addperiod}[1]{
	#1\@addpunct{.}
}
\makeatother

\newtheorem{observation}{Observation}
\newtheorem{claim}{Claim}
\newtheorem*{claim*}{Claim}
\newtheorem{definition}{Definition}
\newtheorem*{definition*}{Definition}
\newtheorem{lemma}{Lemma}
\newtheorem*{lemma*}{Lemma}
\newtheorem{example}{Example}
\newtheorem*{example*}{Example}
\newtheorem{corollary}{Corollary}
\newtheorem*{corollary*}{Corollary}
\newtheorem{theorem}{Theorem}
\newtheorem*{theorem*}{Theorem}
\newtheorem{proposition}{Proposition}
\newtheorem*{proposition*}{Proposition}
\newtheorem{openprob}{Open Problem}
\newtheorem*{openprob*}{Open Problem}

\theoremstyle{definition}  \newtheorem{algdef}{Algorithm}%

\DeclareMathOperator*{\E}{\mathbb{E}}
\newcommand{\reals}{\mathbb{R}}

\newcommand{\Welf}{\mathrm{Welfare}}
\newcommand{\Alg}{\textsc{Alg}}
\newcommand{\Opt}{\textsc{Opt}}
\newcommand{\NP}{\textsf{NP}}
\newcommand{\APX}{\textsf{APX}}

\newcommand{\A}{\mathbb{A}}
\newcommand{\I}{\mathbb{I}}
\newcommand{\Inst}{\mathcal{I}}
\renewcommand{\O}{\mathcal{O}}

\begin{document}

	\title{Matching with Nested and Bundled Pandora Boxes}

	\author{Robin Bowers \\ CU Boulder \and Bo Waggoner  \\ CU Boulder}
	\date{}

	\maketitle

	\begin{abstract}
		We consider max-weighted matching with costs for learning the weights, modeled as a ``Pandora's Box'' on each endpoint of an edge.
		Each vertex has an initially-unknown value for being matched to a neighbor, and an algorithm must pay some cost to observe this value.
		The goal is to maximize the total matched value minus costs.
		Our model is inspired by two-sided settings, such as matching employees to employers.
		Importantly for such settings, we allow for negative values, causing existing approaches to fail.

		We first prove impossibility results for algorithms in two natural classes.
		Any algorithm that ``bundles'' the two Pandora boxes incident to an edge is an $o(1)$-approximation.
		Likewise, any ``vertex-based'' algorithm, which uses properties of the separate Pandora's boxes but does not consider the interaction of their value distributions, is an $o(1)$-approximation.
		Instead, we utilize \emph{Pandora's Nested-Box Problem}, i.e. multiple stages of inspection.
		We give a self-contained, fully constructive optimal solution to the nested-boxes problem, which may have structural observations of interest compared to prior work.
		By interpreting each edge as a nested box, we leverage this solution to obtain a constant-factor approximation algorithm.
		Finally, we show any ``edge-based'' algorithm, which considers the interactions of values along an edge but \emph{not} with the rest of the graph, is also an $o(1)$-approximation.
	\end{abstract}

	\section{Introduction}\label{sec:intro}

We consider a problem of weighted matching under uncertainty about weights, motivated by two-sided matching settings such as job search websites and school admissions.
In many of these matching applications, the values of each side for a potential match are initially unknown until some effort or cost is expended to inspect and discover the value.
Investigating the quality of a school or job may be time consuming, and sorting through applicants may be costly for a company.
The challenge in these settings is balancing the costly acquisition of information against the need to find high-value matches.
We ask whether it is possible for an efficient algorithm to navigate this tradeoff.

Recent work on market design has studied the problem of matching with information acquisition, typically from the perspective of a platform designer in a strategic environment (see Section \ref{subsec:related}).
This paper takes a step back to consider a simpler, purely algorithmic problem.
However, all of our results have practical implications for design of two-sided matching platforms, particularly a number of impossibilities for practically-motivated classes of algorithms.
Both these negative results and our eventual positive results depend on new insights into the \emph{Pandora's Box} model.

\paragraph{Pandora's Box}
We use \citet{weitzman1979optimal}'s Pandora's Box model of costly information acquisition.
In this model, each Pandora box is parameterized by an ``inspection cost'' to open the box and a distribution over the value inside (drawn independently).
In Weitzman's original problem, an algorithm may choose any sequence of inspections, then stop at any time and claim the value inside some opened box.
The goal is to maximize the final value minus all inspection costs paid.

\paragraph{Our model}
In \emph{Pandora's Matching Problem}, we are given a graph with initially uncertain values on edges.
Each edge of the graph has two Pandora boxes, one at each endpoint, representing the values of each of the two parties involved in a potential match.
The weight of an edge is the sum of the values on its endpoints, which may each be inspected separately at a cost.
An algorithm may conduct any sequence of inspections, then stop and output a matching consisting of fully-inspected edges.
The performance, or ``welfare'', of an algorithm is the sum of edge weights in the matching minus the sum of inspection costs paid.
We consider the approximation factor, i.e. worst-case ratio of expected welfare to that of the optimal algorithm (whatever it may be).

Our model captures two-sided preferences in matching markets with costly inspection on both sides.
For example, a company seeking to hire a new worker must conduct time-intensive interviews to determine the fit of a candidate, while the candidate must do the same.
Similarly, a student applying to colleges must invest the time to determine if each school she considers would be a good fit, as applying to every possible school is infeasible.
A university spends time and money processing applications, and cannot give every application a full evaluation.
We are interested in the consequences of this two-sided inspection dynamic for algorithm design.

We will relate our problem to \emph{Pandora's Nested-Box Problem}, in which opening a box uncovers another costly box within, along with some information about the contents.
We give a new treatment of the nested-boxes problem and apply it to address Pandora's Matching Problem.

\paragraph{Negative values}
Prior work of \citet{bowers2023high} implies an approximation algorithm for the \emph{positive-values} setting where all values are positive and exceed inspection costs in expectation.
However, many two-sided matching applications involve negative values.
For example, a candidate generally has a negative value for being assigned a job and needs to be paid a wage in order to do it.
	\footnote{Simple adjustments to reduce to the positive setting do not work.
	E.g., adding the wage of the job to the job's value distorts the social welfare unless we also subtract it from the employers' value for hiring the employee, and in general there will be pairs where the sum is negative.
	Some of our negative results formalize the impossibility of such approaches.}
Because of this, we are particularly motivated by a general-values setting.
\citet{bowers2023high} state obtaining an approximation algorithm in the general-values setting as a main open problem.

\subsection{Results}

We expect the problem to be NP-hard, as even the special case of one Pandora box per edge is suspected to be NP-hard (e.g. \citet{singla2018price}).
However, a $(1/2)$-approximation to that special case is given by \citet{kleinberg2016descending} and \citet{singla2018price}, utilizing a generalization of Weitzman's descending procedure to mimic a greedy matching.
Motivated by this approach, we first consider simple classes of approximation algorithms.

\paragraph{Bundled Box \& Vertex-Based Algorithms}
A first natural approach is to ``bundle'' the two boxes on each edge into one larger Pandora box: if inspecting the edge, always inspect both boxes simultaneously.
Another possible approach is a generalization of the approach taken by \citet{bowers2023high} and \citet{singla2018price}, which we call \emph{vertex-based} algorithms.

Unfortunately, Section \ref{sec:initial-attempts} finds that both bundled box and vertex-based approaches fail in general when applied to the two-sided matching setting.

\begin{theorem*}
	No bundled-box or vertex-based algorithm for Pandora's Matching Problem can guarantee a positive approximation factor.
\end{theorem*}

These impossibilities also have implications for practical platform design.
The bundled setting corresponds to traditional interviewing, where both parties must agree to ``inspect'' together.
This result suggests that welfare may be lost from markets that require participants to coordinate their information-gathering efforts.
Similarly, a vertex-based approach is one that allows for decentralized decisionmaking.
For example, the matching mechanism in \citet{bowers2023high} does not share information about the participants' preferences with each other.
The failure of vertex-based algorithms implies that some information-sharing is necessary to coordinate good matches.

\paragraph{Nested boxes}
Section \ref{sec:multistage} turns to the \emph{Pandora's Nested-Box Problem}, also called Pandora's problem with multiple stages of inspection, to account for interactions between the boxes along an edge.
\citet{kleinberg2016descending} provide a generalization of Weitzman's indices and our descending procedure to the multi-stage inspection setting.
We give a self-contained, constructive treatment with intuition for the algorithm that we believe is of independent value (see Section \ref{subsec:related}).

\paragraph{Fixed-orientation algorithms}
In Section \ref{sec:unordered}, we utilize our nested-boxes machinery to create two approximation algorithms for Pandora's Matching Problem.

\begin{theorem*}[Main result]
	There exists an algorithm which obtains a $1/4$-approximation for Pandora's Matching Problem.
	The guarantee holds even if each edge's pair of values may be arbitrarily correlated with each other.
\end{theorem*}

Both of these algorithms work by first fixing an order of inspection (an ``orientation'')  for each edge, and then running the $1/2$-approximation descending procedure on the fixed-orientation setting.
Finally, we consider simpler deterministic algorithms, and demonstrate that a very natural class of \emph{edge-based fixed-orientation} algorithms also cannot achieve a constant approximation factor.

\begin{theorem*}
	No edge-based fixed-orientation algorithm for Pandora's Matching Problem can guarantee a positive approximation factor.
\end{theorem*}

\subsection{Related Work} \label{subsec:related}

\paragraph{Pandora's Box and its variants}
The Pandora's Box Problem was first introduced by \citet{weitzman1979optimal}.
Our terminology and machinery for the problem follow \citet{kleinberg2016descending}, who used the problem in an auction setting, as well as a one-sided matching setting (i.e. matching with a single Pandora box on each edge).

This approach was extended to many combinatorial optimization problems with information acquisition by \citet{singla2018price}.
Our problem is not addressed by his results in particular because greedy max-weight matching, if there is a weight on each endpoint of an edge, is not a \emph{frugal algorithm}: it needs to consider the contributions of pairs of items (both endpoints of an edge), while a frugal algorithm considers one at a time.
In a sense, our impossibility result for vertex-based algorithms in Section \ref{sec:initial-attempts} implies that no variant of the frugal framework can solve our problem.

For what we call \emph{Pandora's Nested-Box Problem}, also referred to as multiple stages of inspection, there are a number of recent works adjacent to ours.
Most important, the Weitzman index approach and optimal descending algorithm was shown to extend to this case by \citet[Appendix G]{kleinberg2016descending}.
The treatment there is quite general (in particular, handling countably infinite stages of inspection), but it is also quite complex and nonconstructive.
In Section \ref{sec:multistage}, we include a concrete and direct extension of the Weitzman index machinery to the nested-boxes setting.
We present the key features in the main body of the paper because, in addition to enabling our final results, they do provide some novel points of utility and interest:
\begin{itemize}
	\item A simple and direct construction with a much shorter, self-contained treatment.
	Given an example (e.g. a Bernoulli value whose bias is selected over the course of the inspections), it is not necessarily clear from \citet{kleinberg2016descending} how exactly to compute the Weitzman indices; here we give a direct formula.
	\item Illustration of several subtle points that arise in the analysis, such as a precise and constructive formulation of the \emph{non-exposed} condition (also called exercise in the money, claim above the cap, etc.).
	\item Simple and concrete machinery -- Weitzman indices (also called strike prices or reservation values) capped values (also called covered call values or deferred values), and a general ``amortization lemma'' -- that is more amenable to adaptation to different environments.
\end{itemize}
\citet{gupta2019markovian} provides a generalization of \citet{singla2018price}'s frugal algorithms to a multistage inspection model for combinatorial optimization.
They consider optimization over boxes represented by Markov chains with costs for advancement, where a chain must be advanced to a terminal state in order for the box to be added to the solution.
\citet{gupta2019markovian} derives a parallel to the Weitzman index machinery of \citet{kleinberg2016descending} in the discrete Markov decision process setting.
Their approach is essentially a generalization of our Section \ref{sec:multistage}, except with a focus on compatibility with frugal algorithms.
As mentioned above, our two-sided matching problem is incompatible with frugal algorithms, so the results of \citet{gupta2019markovian} do not apply to our main problem.

Several other works study variants of the nested-boxes problem, but they are much farther from our setting.
\citet{ke2019optimal} study a model in which a decisionmaker gathers information about each of a set of alternatives in continuous time for continuous cost, and stops at any time to choose an alternative.
(In particular, full inspection is not obligatory.)
\citet{aouad2020pandora} consider a model with two stages of inspection, partial and full.
However, the algorithm does not have to undertake partial inspection of an option prior to a full inspection.
(Full inspection is still obligatory in order to claim.)
In this setting, \citet{aouad2020pandora} show that a descending index-based policy is no longer optimal, studying features of the optimal solution and approximation algorithms.
\citet{boodaghians2020pandoras} study a somewhat-related setting that can be viewed as Pandora's box problem where boxes must be opened in a particular order, particularly when the ordering follows a linear or DAG structure.

There has been a variety of other recent work on variants of the Pandora's Box Problem, including when inspection is non-obligatory to claim the item~\citep{doval2018whether,beyhaghi2019pandora,beyhaghi2023pandora,fu2023pandora} and when the values in the boxes are correlated~\citep{chawla2020pandora}.
All of the above settings and a number of others are discussed in the recent survey of \citet{beyhaghi2023recent}.

\paragraph{Market design}
A number of recent works in market design or mechanism design have utilized Pandora's box as a model for initially-unknown values.
Closest to our work is \citet{bowers2023high}.
The authors consider essentially the same model as we do, but with each endpoint representing a strategic agent.
They study a descending-price mechanism, an extension of \citet{kleinberg2016descending} proposed by \citet{waggoner2019matching}, in which agents maintain bids on each of their neighbors in the graph and are matched and pay their bids from the highest-total edge downward.
With the assumption of nonnegative values, \citet{bowers2023high} prove a constant-factor Price of Anarchy guarantee, implying that a constant-factor approximation is achievable in that setting efficiently (up to computation of an equilibrium; we extend their result into an efficient algorithm in Appendix \ref{subsec:positive-values-appdx} for completeness).
They give partial results for the general-values case, but leave the existence of a constant-factor approximation as a main open question.

Other works in the market design literature generally focus more on modeling of matching marketplaces than we do here.
\citet{immorlica2021designing} studies design of a platform to facilitate two-sided matching with inspection costs.
There is a continuous population of agents and matching occurs in continuous time, with the platform directing agents of a given type to ``meet'' and inspect counterparts of some other type.
The authors give a $1/4$-approximation to the problem of finding an optimal flow process, despite adhering to incentive-compatibility constraints.
\citet{immorlica2020information} and \citet{hakimov2023costly} both study the college application and admissions process with costly information acquisition through a theoretical lens; \citet{immorlica2020information} include a case where each student is modeled as a Pandora box (two-sided matching with one-sided information acquisition).
Other works such as \citet{ashlagi2020clearing} study the cost of learning preferences through a lens of communication complexity.
\citet{chade2017sorting} surveys prior approaches in the economics literature to matching markets with information acquisition, but the models and techniques are quite different from our problem.

	\section{Preliminaries}\label{sec:prelims}

In this section, we introduce notation and define the main problem studied in this paper, Pandora's Matching Problem.
We also introduce tools from prior work for Pandora's box analysis.

\subsection{Model \& Notation}

Recall that in the traditional Pandora's Box Problem, Pandora is presented with $n$ boxes.
Formally, a \emph{Pandora box} is a pair $(D,c)$ consisting of the distribution $D$ of the value $v \in \reals$ inside the box along with the cost $c \in \reals_{\geq 0}$ to open the box and observe the value.
Pandora can pay the costs and open as many boxes as she wishes, but can only claim the value from one box.

To adapt the problem to a two-sided matching setting, we let each edge $\{i, j\}$ in a graph have a different Pandora box associated with each endpoint: vertex $i$ has some initially-unknown value for matching to $j$, and vice versa.
Formally, the input to the \emph{Pandora's Matching Problem} consists of:
\begin{itemize}
	\item A graph $G = (V, E)$. We do not require the graph to be bipartite.
	\item For each edge $\{i,j\} \in E$: a pair of Pandora boxes $(D_{ij},c_{ij})$ and $(D_{ji},c_{ji})$.
\end{itemize}
An algorithm iteratively chooses an unopened box $(i,j)$, i.e. an endpoint of some edge of the graph, and ``opens'' it by paying the cost $c_{ij}$, learning $v_{ij}$.
Each $v_{ij}$ is drawn from its respective distribution $D_{ij}$ independently.
At any point, the algorithm may choose to stop inspecting and output a matching of $G$ consisting of fully-inspected edges.
Note that we distinguish the edge $\{i, j\}$ from the two boxes on its endpoints, $(i, j)$ and $(j, i)$.

Given an algorithm and an input instance $\Inst{} = (G, \{D_{ij}\}, \{c_{ij}\})$, we let $\I_{ij}$ be the indicator variable that box $(i,j)$ is opened (or ``inspected''), and let $\A_{ij}$ be the indicator variable that box $(i,j)$ is included in the matching (or ``claimed'').
Formally, the matching requirement is $\A_{ij} = \A_{ji} ~~ (\forall \{i,j\} \in E)$ along with $\sum_{j \in V} \A_{ij} \leq 1 ~~ (\forall i \in V)$.

We impose the classic ``obligatory inspection'' requirement, that both endpoints of an edge must be inspected before the edge can be included in the matching.
Formally, $\A_{ij} \leq \I_{ij}$ for each box $(i,j)$.
The performance of algorithm \Alg{} on the instance $\Inst{}$ is
\begin{align*}
	\Welf(\Alg{}) := \E\left[\sum_{(i, j)} \A_{ij}v_{ij}- \I_{ij}c_{ij}\right],
\end{align*}
which we refer to as the \emph{welfare} of the algorithm.\footnote{The ``welfare'' terminology (i) references our motivating economic applications and (ii) differentiates from the term ``value'' which refers to the contents of Pandora boxes.}
We define the \emph{welfare contribution} of box $(i,j)$ to be $ \A_{ij}v_{ij}- \I_{ij}c_{ij}$.
We denote the optimal algorithm by \Opt{}.
We say an algorithm \Alg{} achieves an approximation factor $\alpha$ or is an $\alpha$-approximation if, for all instances $\Inst$,
\begin{align*}
  \frac{\Welf(\Alg{})}{\Welf(\Opt{})} &\geq \alpha .
\end{align*}

\subsection{Pandora's Toolbox}
We recall several tools from the existing Pandora's Box literature.
The optimal algorithm for the original problem of \citet{weitzman1979optimal} relies on indices that we will call $\sigma_{ij}$ for each box $(i,j)$.
\begin{definition}[Weitzman Index \& Capped Value]\label{def:weitz-index}
	The \emph{Weitzman index} of a Pandora box $(D_{ij},c_{ij})$ is the unique $\sigma_{ij}$ such that
	\begin{align*}
		\E_{v_{ij} \sim D_{ij}}\left[(v_{ij} - \sigma_{ij})^+\right] &= c_{ij},
	\end{align*}
	and the \emph{capped value} of the box is the random variable $\kappa_{ij} = \min(v_{ij}, \sigma_{ij})$.
\end{definition}
Here $(\cdot)^+ := \max\{\cdot, 0\}$.
Discussions, such as a financial-options interpretation, can be found in other works such as \citet{kleinberg2016descending,beyhaghi2023recent}.
\begin{definition}\label{def:non-exposed}
	An algorithm is \emph{exposed} on box $(i,j)$ if there is a nonzero probability that $\I_{ij} = 1$, $v_{ij} > \sigma_{ij}$, and $\A_{ij} = 0$.
	Otherwise, the algorithm is \emph{non-exposed} on box $(i,j)$.
\end{definition}
In other words, an algorithm is exposed on a box if it opens the box, finds that the value is larger than the index, and ultimately fails to claim that box.
We call an algorithm \emph{non-exposed} if it is non-exposed on all boxes.

\begin{lemma}[ \cite{kleinberg2016descending}] \label{lemma:KWW-key}
	For any algorithm, for any Pandora box $(i,j)$,
	\[ \E\left[ \A_{ij} v_{ij} - \I_{ij} c_{ij} \right] \leq \E \left[ \A_{ij} \kappa_{ij} \right], \]
	with equality if and only if the policy is non-exposed on box $(i,j)$.
\end{lemma}
In the classic Pandora's Box Problem, recall that there are $n$ Pandora boxes and Pandora must ultimately select just one.
As discussed in \citet{kleinberg2016descending}, Lemma \ref{lemma:KWW-key} implies optimality of a simple \emph{descending} procedure, which inspires approaches in this work.
In the descending procedure, boxes are opened in order of their index from largest to smallest.
If at any point the largest value observed so far is larger than all remaining indices, the search is stopped and that value is claimed.
This algorithm is non-exposed.
Therefore, summing the equality in Lemma \ref{lemma:KWW-key} over all boxes, its welfare is equal to its total expected capped value.
But it always claims the box with the highest capped value.
Lemma \ref{lemma:KWW-key} implies that this maximizes the possible welfare of any policy that claims at most one box, so the descending procedure is optimal.

	\section{Initial Attempts}\label{sec:initial-attempts}

We first consider two natural classes of algorithms with intuitive interpretations, and demonstrate that neither of them can achieve welfare guarantees in the general-values setting.

\subsection{Bundled-Box Algorithms}\label{subsec:bundled-box}

\citet{singla2018price} provides a general welfare approximation guarantee for the one-sided matching setting, in which each edge has only one Pandora box on it.
Can our two-box setting be reduced to this setting, by simply forcing boxes on an edge to always be opened simultaneously?
In many real-world settings such as interviewing, where a company and employee must both invest simultaneously to reveal their values, it is natural to require the algorithm to ``bundle'' the two boxes incident to an edge.
Formally, a \emph{bundled-box algorithm} is one that, if it opens one box on an edge, always opens the other next.
In particular, it satisfies $\I_{ij} = \I_{ji} ~~ (\forall \{i,j\} \in E)$.

A bundled-box algorithm functionally reduces each edge to a single Pandora box, with value the sum of individual values and cost the sum of costs.
We model the bundled box on edge $\{i,j\}$ as having a combined value $v_{\{i,j\}}' = v_{ij} + v_{ji}$, and we define the new value distribution $v_{\{i,j\}}'\sim D_{\{i,j\}}'$, the convolution of $D_{ij}$ and $D_{ji}$.

If we let $\Opt'$ denote the optimal algorithm for this setting, then \citet{kleinberg2016descending} and \citet{singla2018price} both give $1/2$-approximations to $\Opt'$, extending Weitzman's descending procedure: inspect the edges from largest index downward, matching any edge once it has been inspected and its value is larger than all remaining indices, and removing edges from consideration once they are not legal to add to the matching.

\paragraph{The power of bundled boxes}
The key question is then whether $\Opt'$ can approximate $\Opt$, i.e. whether bundling comes at a high cost.
On one hand, $\Opt$ clearly has significant additional power over simultaneous inspection in choosing inspections.
On the other, edges are ultimately bundled when claimed by a matching.
Here, we show that bundling the edges can perform arbitrarily poorly relative to algorithms that inspect asynchronously.

\begin{proposition} \label{prop:simul-fails}
  No bundled-box algorithm can guarantee better than an $O(\tfrac{1}{n})$-approximation factor for Pandora's Matching Problem on $n$ vertices.
\end{proposition}

 The proof of Proposition \ref{prop:simul-fails} uses the bundled star graph given in Figure \ref{fig:star-graph}.
 The key observation is that a bundled-box algorithm is forced to pay to open boxes with unhelpful information, which can be avoided by an algorithm inspecting endpoints independently.
 This result may have implications for the design of matching platforms such as labor markets, where it suggests enabling at least some asynchronous inspection.

 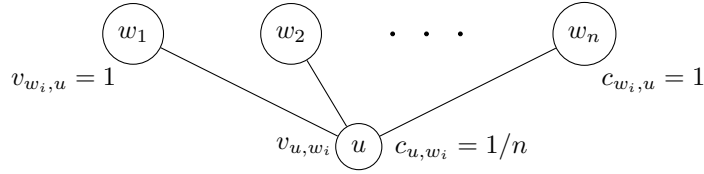
\begin{figure}
 	\begin{center}
 			\begin{tikzpicture}[
		wide/.style={line width=4pt},
		every node/.style={circle,draw, minimum size = 1},scale=1.5]
		\node (w1) at (-2,0) {\small $w_1$};
		\node (w2) at (-.6,0) {\small $w_2$};
		\node (wn) at (2,0) {\small $w_n$};

		\node (u) at (0,-1) {\small $u$};

		\node[circle,fill,minimum size=0.05,inner sep=0.5pt,] (dots2) at (0.6,0) {};
		\node[circle,fill,minimum size=0.05,inner sep=0.5pt,] (dots1) at ($(dots2) - (.3,0)$) {};
		\node[circle,fill,minimum size=0.05,inner sep=0.5pt,] (dots3) at ($(dots2) + (.3,0)$) {};

		\draw (u) -- (w1);
		\draw (u) -- (w2);
		\draw (u) -- (wn);

		\node[below left, draw=none,fill=none](vwu) at ($(w1) - (.2,0)$) {\small $v_{w_i,u} = 1$};
		\node[left, draw=none,fill=none](vuw) at ($(u) - (.1,0)$) {\small $v_{u,w_i}$};
		\node[below right, draw=none,fill=none](cwu) at ($(wn) + (.2,0)$) {\small $c_{w_i,u} = 1$};
		\node[right, draw=none,fill=none](cuw) at ($(u) + (.2,0)$) {\small $c_{u,w_i} = 1/n$};
	\end{tikzpicture}
 	\end{center}
 	\caption{
 		A reproduction of Figure \ref{fig:star-graph}.
 		The ``bundled star graph'' with vertices $w_{1}, \dots, w_{n}$ surrounding $u$.
 		Each leaf $w_i$ has a constant value and cost of $1$.
 		The cost for $u$ to inspect any of its values is $c_{u,w_i} = 1/n$, and the value is either $v_{u,w_{i}} = n$ with probability $1/n$ or 0 otherwise.
 		A bundled-box algorithm performs arbitrarily worse on this example than an algorithm which can inspect endpoints independently.
 	}
 	\label{fig:star-graph}
 \end{figure}

 \begin{proof}
 	We combine the values along each edge of the bundled star graph and note that exactly one edge can be chosen by a matching: we can claim exactly one box.
 	We denote the combined boxes on each edge $\{i,j\}$ as having a value $v_{\{i,j\}}' = v_{ij} + v_{ji}$ with new value distribution $v_{\{i,j\}}'\sim D_{\{i,j\}}'$, the convolution of $D_{ij}$ and $D_{ji}$.

 	Using Lemma \ref{lemma:KWW-key} from \cite{kleinberg2016descending}, the expected value of the optimal algorithm is
 	\begin{align*}
 		\E\left[ \sum_{i}\A_{\{u,w_i\}}v_{\{u,w_i\}}' - \I_{\{u,w_i\}}\left(c_{\{u,w_i\}} + c_{\{u,w_i\}}\right)\right] &= \E\left[(\max_{i} \kappa_{\{u,w_i\}})^+\right],
 	\end{align*}
 	where $\kappa_{\{u,w_i\}} = \min(\sigma_{\{u,w_i\}}, v_{\{u,w_i\}}')$ and $\sigma_{\{u,w_i\}}$ is the index for the combined boxes.
 	We can compute theindexfor every edge of the example as
 	\begin{align*}
 		\frac{1}{n} + 1 &= \E\left[(v_{u, w_i} + 1 - \sigma_i)^+\right] \\
 		&= \frac{1}{n}( n + 1 - \sigma_i) + \left(1 - \frac{1}{n}\right) ( 1 - \sigma_i)  \\
 		&= 2 - \sigma_i  \\
 		&\implies ~~ \sigma_i = 1 - \tfrac{1}{n} .
 	\end{align*}
 	By Lemma \ref{lemma:KWW-key} then, the optimal bundled-box algorithm's welfare is at most $\E[(\max_i \min(\sigma_{\{u,w_i\}}, v_{\{u,w_i\}}'))^+] \leq 1$.

 	However, with independent inspections allowed, a simple algorithm does much better: inspect every edge from $u$'s side.
 	If for any $i$ the realized value is  $v_{u, w_i} = 1/n$, then inspect $v_{w_i, u}$ and match that edge.
 	This algorithm has expected welfare $(1-(1-1/n)^n)(n + 1 - 1) - 1 \geq n(1-e^{-1}) - 1 = \Omega(n)$.
 \end{proof}

\subsection{Vertex-Based Algorithms}\label{subsec:vertex-based}

We next consider algorithms which treat the boxes independently of one another, making decisions to open or permanently accept boxes independently of any other boxes.
A natural approach is to directly extend \citet{weitzman1979optimal}'s optimal index-based approach to this setting.
For example, one could inspect the boxes from largest Weitzman index downward, matching an edge once both of its endpoints are inspected and the sum of the values is large enough.
Approaches like this match the \emph{frugal algorithm} framework of \citet{singla2018price}.
We find that, when values must be positive, the mechanism of \citet{bowers2023high} yields a $1/4$-approximation algorithm, but that no vertex-based algorithm can achieve a constant approximation for general values.

\begin{definition}\label{def:vertex-based}
	An algorithm is \emph{vertex-based} if, for each ordered pair $(i, j)$, the algorithm uses only Weitzman's index $\sigma_{ij}$, the inspection cost of a vertex $c_{ij}$, and (when revealed) the value $v_{ij}$, but does not have access to the distribution $D_{ij}$.
\end{definition}

To enable analysis, Definition \ref{def:vertex-based} formalizes vertex-based algorithms by limiting the information available to the algorithm.
A vertex-based algorithm knows everything about the setting except for the distribution of values $D_{ij}$ at each endpoint.
We replace the knowledge of $D_{ij}$ with direct access to the Weitzman index $\sigma_{ij}$ instead.

\paragraph{Positive Values}
In the \emph{positive-values} setting of Pandora's Matching Problem, value distributions $D_{ij}$ are supported on the positive reals and $\E[v_{ij}] \geq c_{ij}$ for all $(i,j)$.
This implies $\sigma_{ij}\geq 0$ and $\kappa_{ij}\geq 0$ as well.

This restricted setting is already nontrivial.
For example, our impossibility result of Proposition \ref{prop:simul-fails} holds even in the positive-valued setting, showing that no bundled-box algorithm can obtain a better-than-$O(1/n)$ approximation ratio.
However, there is a vertex-based algorithm which achieves a $1/4$-approximation to the optimal expected welfare in the positive-values setting.
This is a corollary of the results of \citet{bowers2023high}, which gives a descending-price style mechanism that achieves a constant Price of Anarchy by focusing on the higher-capped-value endpoint of each edge.
In Appendix \ref{sec:separate-models-appdx}, for completeness, we extract a simple approximation algorithm.

\begin{proposition}\label{prop:pos-approx}
	There exists a vertex-based algorithm which achieves a $1/4$-approximation in the positive-values setting of Pandora's Matching Problem.
\end{proposition}

\paragraph{General Values and Failure of Vertex-Based Algorithms}
However, this result hinges on the nonnegativity of values: it will match an edge when one endpoint has a positive value, even if the other endpoint's value is arbitrarily negative.
We find that this problem is not fixable for any vertex-based algorithm when values may be negative.

\begin{proposition} \label{prop:vertex-based-fails}
	No vertex-based algorithm can guarantee any positive approximation factor for Pandora's Matching Problem.
\end{proposition}

The proof of Proposition \ref{prop:vertex-based-fails} is in Appendix \ref{subsec:indistinguishable-appdx}, and relies on a single-edge counterexample.
In this example, a vertex-based algorithm cannot distinguish the endpoints to pick an inspection ordering but an optimal policy must inspect them in a specific order.

This negative result has implications for platform or mechanism design: an algorithm must coordinate a nontrivial amount of information-sharing in order to guarantee high welfare in general.
We must turn to algorithms which consider edges as a whole, rather than the individual indices of the endpoints.

\paragraph{The correct way to combine endpoints' information}
If we must combine endpoints' information to make good decisions, how should we do so?
An edge in the Pandora's Matching Problem involves \emph{multiple stages} of information gathering: first we inspect one endpoint and learn something about the edge's total value, and later we inspect the other.
Therefore, we must put the matching problem on hold and investigate \emph{Pandora's Nested-Box Problem}.
Once we have a solution to that problem, we will come back and apply it in the matching setting.

	\section{Pandora's Nested-Box Problem}\label{sec:multistage}

To solve the general problem of matching with vertex-based Pandora boxes, we will need a new approach.
We turn to the model of \emph{Pandora's Nested-Box Problem}.
In this model, when Pandora first pays the cost to open a box, there is a second box inside, with its own associated cost; and so on, with the final reward inside the last inner box.

While we are coining the term ``Nested Boxes'' to refer to this problem, it has previously been studied and an optimal solution is known~\citep{kleinberg2016descending}.
In fact, just as with the standard Pandora's Box Problem, this follows directly from the Gittins Index Theorem as a special case~\citep{gittins1979bandit}.
However, just as with the standard Pandora's Box Problem, a constructive solution can give insights about the structure of the problem.

\subsection{Model}

For clarity, we will introduce terminology that distinguishes between \emph{baskets} and \emph{boxes} (illustrated in Figure \ref{fig:nested-baskets}).
Using this terminology, in the traditional Pandora's Box Problem, Pandora is presented with $n$ baskets, and each basket initially contains a single box.
When Pandora chooses to inspect the box in basket $i$, she pays the cost $c_i$ to open its box, finds a reward $v_i$, throws away the box packaging, and places $v_i$ in the basket.
In the end, Pandora can pick one basket with a reward in it to take home.
We modify this problem into the \emph{Pandora's Nested-Box Problem}.

\paragraph{Pandora's Nested-Boxes}
A \emph{Pandora basket} $i$, defined by a pair $(D_i,c_i)$, contains a nested set of $d \geq 1$ boxes.
For $\ell=1,\dots,d-1$, opening the $\ell$th box in basket $i$ provides the algorithm with a signal $s_{i}^{(\ell)}$ regarding the value $v_i$ of the basket.
Opening the $d$th box reveals $v_i$.
The signals and final value are drawn jointly from a known distribution, i.e. $(s_i^{(1)}, \dots, s_i^{(d-1)}, v_i) \sim D_i$.
The sequence of costs $c_i = (c_i^{(1)}, \dots, c_i^{(d)})$ to perform each inspection is also known.

We refer to the most-recently opened box $\ell$ in a basket as the current \emph{stage} of a basket: When we have opened no boxes, we are in stage zero.
After opening all $d$ boxes, we are on the $d$th stage, and after claiming the item, the $d+1$st.
We modify the indicator variables $\A$ and $\I$ to accommodate nested boxes.
$\A_i$ continues to denote the claiming of basket $i$, but now the indicator variable for opening the $\ell$th box in basket $i$ is $\I_i^{(\ell)}$.
We will let $\I_i^{(d+1)} = \A_i$ for notational convenience in inductive statements.
We refer to opening a nested box or claiming the final item as \emph{advancing} the stage of the basket.

All previous boxes $k=1,\dots,\ell-1$ in the basket must be opened before opening nested box $\ell$, and all $d$ boxes in a basket must have been opened to claim the contents of the basket.
Formally, we require for any algorithm that $\I_i^{(1)} \geq \cdots \geq \I_i^{(d)} \geq \I_i^{(d+1)} = \A_i$.

\begin{figure}
	\begin{center}
	\includegraphics[width=3cm]{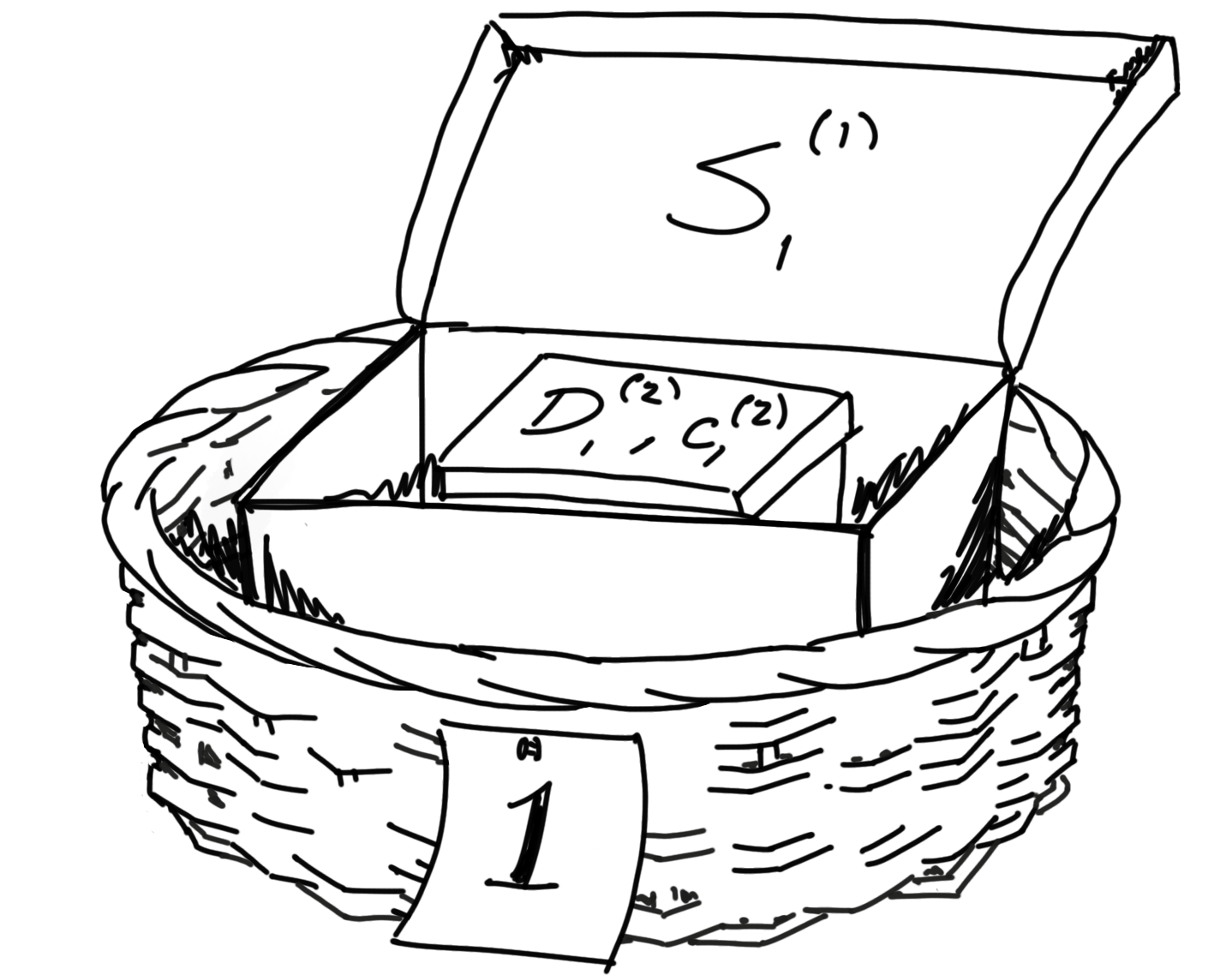}
	\includegraphics[width=3cm]{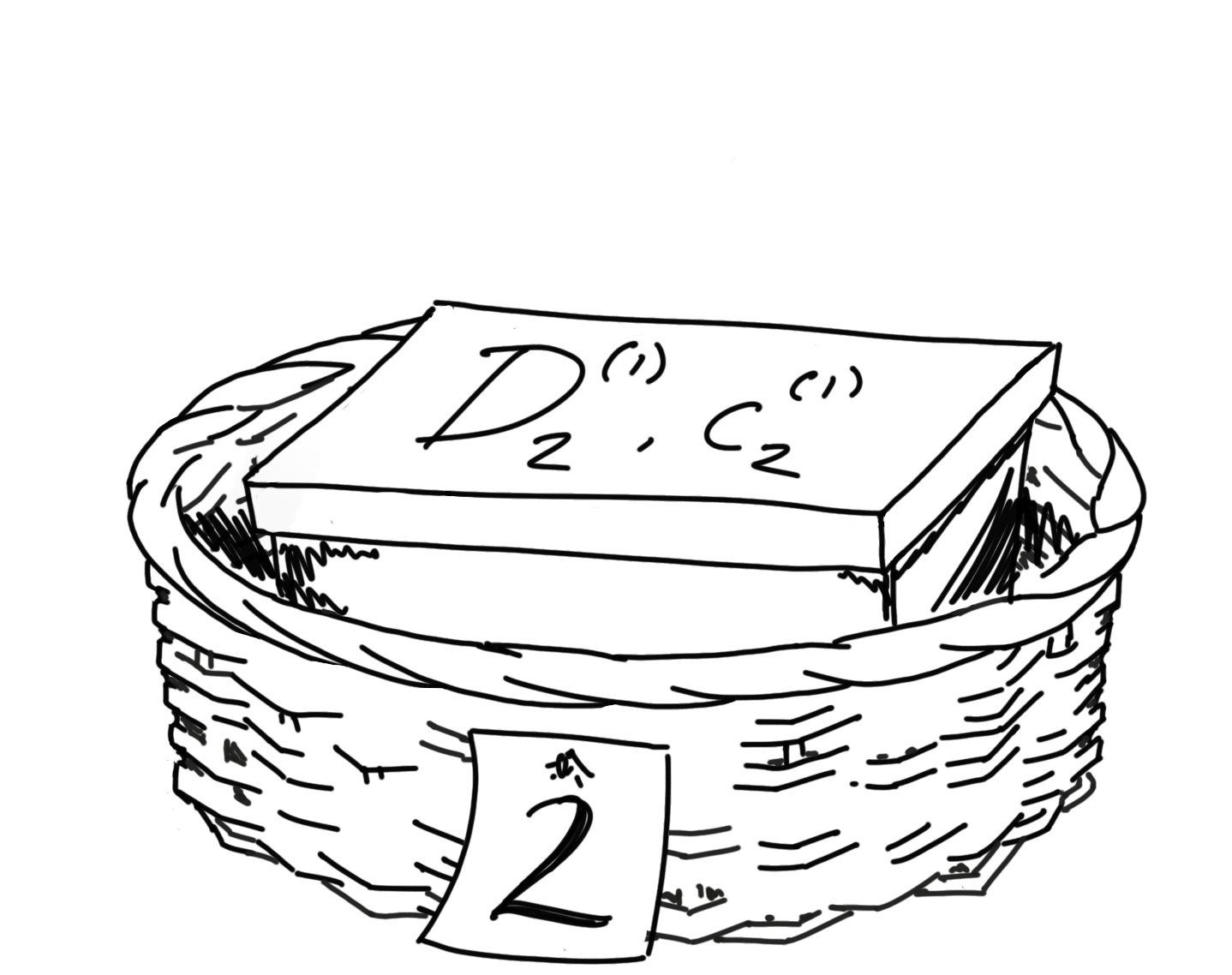}
	\includegraphics[width=3cm]{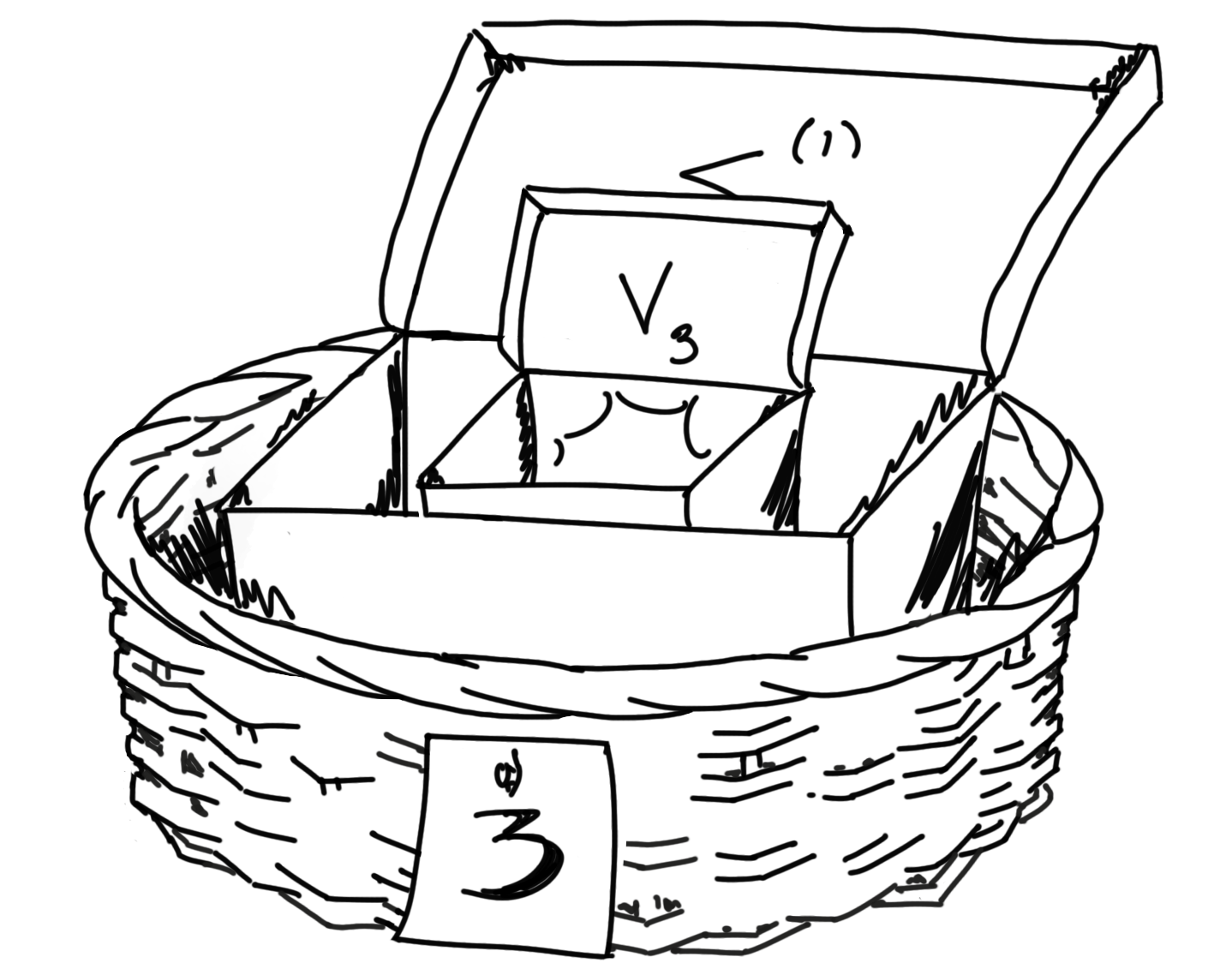}
	\caption{$n = 3$ Pandora baskets with $d = 2$ nested boxes, advanced to various stages with signals and/or values revealed.
	All distributions and costs are known ahead of time.}
	\label{fig:nested-baskets}
	\end{center}
\end{figure}

\paragraph{Optimization problems}
In the spirit of \citet{singla2018price}, we allow for general optimization problems.
In such a problem we are given a set of Pandora baskets, each specified by $(D_i,c_i)$, along with a constraint on which subsets may be claimed.
The algorithm performs any sequence of steps in which it selects a basket and advances it, and eventually stops and claims a subset of fully-advanced baskets that satisfies the constraint.
In particular, in the Pandora's Nested-Box Problem, the algorithm may claim at most one basket.
The \emph{welfare contribution} of basket $i$ is $\A_iv_i - \sum_{\ell = 1}^d \I_i^{(\ell)}c_i^{(\ell)}$.
The expected welfare of an algorithm is the sum of the expected welfare contributions of the baskets.

\paragraph{Remarks}
The model and results in this section can be generalized so that the ``depth'' $d$ of each basket may be a different integer, and so that the costs of inner boxes are random variables, with each $c_i^{(\ell)}$ a function of the signals $s_i^{(1)},\dots,s_i^{(\ell-1)}$.
Essentially no change to the below analysis is required.
The case of possibly-infinite depth is treated in \citet{kleinberg2016descending}.

\subsection{Pandora's Toolbox for Nested Boxes}\label{subsec:pandora-toolbox}
Classic Pandora's Box analyses depend on the proper definitions of the index and capped value, as well as on the key lemma (Lemma \ref{lemma:KWW-key}).
We need to generalize these definitions correctly and prove the analogous lemma.
\citet{kleinberg2016descending} give more general but nonconstructive definitions, and fleshing out these details reveals a few interesting subtleties and useful intuitions.

Conceptually, when defining an index $\sigma_i^{(\ell)}$ for the nested box $\ell$, we would like to capture the usefulness of its contents, box $\ell+1$, by some ``proxy value'' $\hat{v}_i^{(\ell+1)}$, such that the index could be assigned in the classical way as the solution to $\E[(\hat{v}_i^{(\ell+1)} - \sigma_i^{(\ell)})^+] = c_i^{(\ell)}$.
In fact, this turns out to be possible, where the ``proxy value'' is none other than the capped value of box $\ell+1$.
This allows us to define $\kappa_i^{(\ell)}$ and $\sigma_i^{(\ell)}$ together via backward induction.

\begin{definition}[Nested Weitzman Index \& Capped Value] \label{def:nested-index}
	The \emph{nested Weitzman index} for box $i$ at the $\ell$th stage of inspection is the random variable $\sigma_i^{(\ell)}$, a function of $s_i^{(1)},\dots,s_i^{(\ell-1)}$, that uniquely solves
	\begin{equation*}
		\E\left[\left(\kappa_i^{(\ell+1)} - \sigma_i^{(\ell)}\right)^+~\vline~s_i^{(1)},\dots,s_i^{(\ell-1)}\right] = c_i^{(\ell)},
	\end{equation*}
	where $\kappa_i^{(\ell)}$ is the \emph{capped value} $\kappa_i^{(\ell)}=\min\{ \sigma_i^{(\ell)},\kappa_i^{(\ell+1)} \}$ for $1\leq\ell\leq d$, and we define $\kappa_i^{(d+1)} = v_i$.
	We also define $\sigma_i^{(d+1)} = v_i$.
\end{definition}

\subsubsection{Nested non-exposure}
We also adapt the notion of non-exposure to fit the nested box setting.
Intuitively, potential for exposure occurs when the algorithm opens a box $\ell$ in basket $i$ and receives a ``good news'' signal, resulting in a high Weitzman index for the next box.
To avoid exposure, the algorithm should continue by opening the next box in this basket.
Given our claim that $\kappa_i^{(\ell)}$ represents the value of the $\ell$th box, we might like to define exposure as a case where $\kappa_i^{(\ell+1)} > \sigma_i^{(\ell)}$, and the algorithm has advanced basket $i$ to stage $\ell$, yet the algorithm does not advance it to stage $\ell+1$.
Consider a case, however, where $\sigma_i^{(1)} < \kappa_i^{(3)} < \sigma_i^{(2)}$ and we have advanced to stage $2$.
In this case, not advancing to stage $3$ would cause an exposure \emph{relative to the decision to advance to stage $2$}, because $\kappa_i^{(3)} > \sigma_i^{(1)}$, even if there is no exposure relative to $\sigma_i^{(2)}$.
This intuition motivates our definition of non-exposure.

Let $\gamma_i^{(\ell)} = \min_{1\leq\ell'\leq \ell}\sigma_i^{(\ell')}$, the minimum of the first $\ell$ Weitzman indices for basket $i$.
We note that $\gamma_i^{(d+1)} = \kappa_i^{(1)} = \min_{\ell} \sigma_i^{(\ell)}$.
See Figures \ref{fig:gamma-kappa} and  \ref{fig:graph} for a visual depiction.

\begin{figure}
  \begin{center}
    \[ \rlap{$\overbrace{\phantom{\sigma_i^{(1)} \quad\quad \sigma_i^{(2)} \quad \cdots \quad \sigma_i^{(\ell)}}}^{\text{minimum } = ~ \gamma_i^{(\ell)}}$}
       \sigma_i^{(1)} \quad\quad \sigma_i^{(2)} \quad \cdots \quad \underbrace{\sigma_i^{(\ell)} \quad \cdots \quad \sigma_i^{(d)} \quad\quad \sigma_i^{(d+1)}\hspace{-0.8ex}=v}_{\text{minimum } = ~ \kappa_i^{(\ell)}} \]
  \end{center}
  \caption{The definitions of $\kappa_i^{(\ell)}$ and $\gamma_i^{(\ell)}$ in terms of the nested indices $\sigma_i^{(\ell)}$ of a Pandora basket $i$.
           Note that $\gamma_i^{(1)}, \gamma_i^{(2)}, \dots$ is a weakly decreasing sequence, while $\kappa_i^{(1)}, \kappa_i^{(2)}, \dots$ is weakly increasing (see Figure \ref{fig:graph}).
       }
  \label{fig:gamma-kappa}
\end{figure}

\begin{definition}\label{def:nested-non-exposed}
	An algorithm is \emph{exposed} on basket $i$ if there is nonzero probability that there exists $\ell \in \{1,\dots,d\}$ such that $\I_i^{(\ell)} > \I_i^{(\ell+1)}$ but $\gamma_i^{(\ell)} < \kappa_i^{(\ell+1)}$.
	Otherwise, it is \emph{non-exposed} on basket $i$.
	An algorithm is non-exposed if it is non-exposed on all baskets.
\end{definition}

Definition \ref{def:nested-non-exposed} raises an apparent obstacle.
In the classic Pandora's Box problem, it is obvious how to keep an algorithm non-exposed: if the value is higher than the index, claim it.
But here, at an arbitrary stage $\ell$, $\kappa_i^{(\ell+1)}$ is not visible to the algorithm, as its realization depends on what would happen if the basket continued to advance.
Therefore, it is not obvious how to guarantee non-exposure.
The resolution comes from observing that $\kappa_i^{(\ell+1)} \leq \sigma_i^{(\ell+1)}$.
Therefore, we can prevent exposure by always continuing to advance if the next box's Weitzman index is larger than $\gamma_i^{(\ell)}$, the minimum of the indices so far.
It is necessary to compare to $\gamma_i^{(\ell)}$ rather than $\sigma_i^{(\ell)}$ by the previous discussion.

\subsubsection{The key lemma}
In the classic setting, the key lemma (Lemma \ref{lemma:KWW-key}) upper-bounds the expected welfare contribution of a Pandora box, with equality iff non-exposure.
It can also be referred to as an amortization lemma, as it relates the average welfare to the average capped value.
In our nested setting, the analogue is the following.

\begin{lemma}\label{lemma:nested-key}
	For any algorithm, the expected welfare contribution of a Pandora basket $i$ satisfies
	\begin{align*}
		\E\left[\A_iv_i - \sum_{\ell = 1}^d \I_i^{(\ell)}c_i^{(\ell)}\right] \leq \E\left[\A_i\kappa_i^{(1)}\right],
	\end{align*}
	with equality if and only if the algorithm is non-exposed on basket $i$.
\end{lemma}

We prove Lemma \ref{lemma:nested-key} in Appendix \ref{sec:nested-appdx}.

\begin{figure}
	\begin{center}
\begin{tikzpicture}[scale=1,domain=0:8,
	declare function={
		sig1(\x)= 2.75;
	},
	declare function={
		sig2(\x)= 2;
	},
	declare function={
		sig3(\x)= 3;
	},
	declare function={
		sig4(\x)= 1.75;
	},
	declare function={
		sig5(\x)= 1;
	},
	declare function={
		sig6(\x)= 2.5;
	},
	declare function={
		sig7(\x)= 1.5;
	},
	declare function={
		sig8(\x)= 2.25;
	},
	declare function={
		gamma(\x)= (\x < 1) * (2.75) +
		and(1 <= \x, \x < 3) * (2)  +
		and(3 <= \x, \x < 4) * (1.75)  +
		(4 <= \x) * (1)
	;
	},
	declare function={
	kappa(\x)= (\x < 5) * (1) +
	and(5 <= \x, \x < 7) * (1.5)  +
	(7 <= \x) * (2.25)
	;
	}
	]

	\begin{axis}[
		axis lines=middle,
		ymin=-0.2, ymax=4,
		xmin=-0.2, xmax=8.3,
		ymajorticks=false,
		xtick={0,...,7},
		y label style={at={(axis description cs:-0.05,0.5)},rotate=90,anchor=south},
		x label style={at={(axis description cs:0.5,-0.05)},anchor=north},
		ylabel=index,
		xlabel=stage $\ell$,
		xticklabel=\empty,
		domain=0:8,samples=101,
		axis equal image
		]

		\pic at (1.8,-.19) {myrec={1}{3.5}{fill=yellow!12.5!white}};
		\pic at (4.8,-.19) {myrec={3}{3.5}{fill=yellow!12.5!white}};

		\addplot [black,very thick,domain=0:1] {sig1(x)};
		\addplot [black,very thick,domain=1:2] {sig2(x)};
		\addplot [black,very thick,domain=2:3] {sig3(x)} node[above] {$\sigma_i^{(\ell)}$};
		\addplot [black,very thick,domain=3:4] {sig4(x)};
		\addplot [black,very thick,domain=4:5] {sig5(x)};
		\addplot [black,very thick,domain=5:6] {sig6(x)};
		\addplot [black,very thick,domain=6:7] {sig7(x)};
		\addplot [black,very thick,domain=7:8] {sig8(x)};

		\addplot [dashdotted, magenta, thick,domain=8:0] {kappa(x)} node[above right] {$\kappa_i^{(\ell)}$};

		\addplot [dashed, cyan, thick] {gamma(x)}  node[above] {$\gamma_i^{(\ell)}$};
	\end{axis}

\node [color=orange] at (5.4,0.5) {\tiny opportunity for exposure};

\end{tikzpicture}
\end{center}

	\caption{
		One realization of the Weitzman indices $\sigma_i^{(1)},\dots$ when opening nested boxes in a Pandora basket $i$.
		The minimum index so far is $\gamma_i^{(\ell)}$ and the minimum yet to come (including the final value) is the capped value $\kappa_i^{(\ell)}$.
		Yellow regions represent possible opportunities for exposure: if the index $\sigma_i^{(\ell)}$ of the next box under consideration is higher than the current $\gamma_i^{(\ell-1)}$, it is possible that $\kappa_i^{(\ell)}> \gamma_i^{(\ell - 1)}$, in which case failing to continue opening would cause exposure.
	}
	\label{fig:graph}
\end{figure}
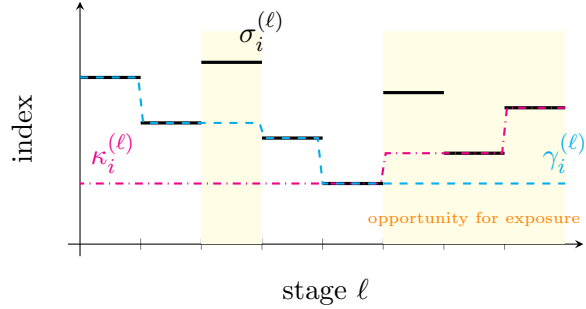

It is not immediately clear that non-exposed algorithms exist, as the conditions in Definition \ref{def:nested-non-exposed} rely on information unavailable to a decision-maker at the time of each choice.

We discuss this tension further in Appendix \ref{sec:nested-appdx}, and provide a generalization of Lemma \ref{lemma:KWW-key} for the nested-boxes setting.
We then use these tools to construct a class of non-exposed algorithms, which we call \emph{descending procedures}.

\subsection{Descending Procedures} \label{sec:opt-nested-policy}
Recall that \emph{advancing} a basket refers to opening the next box or claiming the item if all boxes are opened.

\begin{definition}\label{def:descending-nonexposed}
	Call an algorithm a \emph{descending procedure} if: \emph{(a)} it maintains a non-increasing subset of baskets, the \emph{eligible} set, \emph{(b)} at each step it advances the eligible basket with the largest Weitzman index, and \emph{(c)} it never removes a basket from the eligible set in the same step in which the basket was advanced.
\end{definition}

The idea of the eligible subset is that in optimization problems such as matching, the constraints often cause an algorithm to permanently rule out some basket from further consideration.

\begin{lemma} \label{lemma:descending-nonexposed}
	Any descending procedure is non-exposed.
\end{lemma}

We will crucially use Lemma \ref{lemma:descending-nonexposed}, proved in Appendix \ref{sec:nested-appdx}, to show our main positive result for Pandora's Matching Problem.

As a warm-up to matching, we sketch a proof that the natural descending algorithm is optimal for Pandora's Nested-Box Problem, i.e. the case where Pandora may select any one box.
The descending algorithm always advances the basket $i$ with maximum current Weitzman index $\sigma_i^t$, until one has been claimed or all indices are negative.
Its capped value is pointwise maximum, i.e. $\sum_i \A_i \kappa_i^{(1)} = \max_i ( \kappa_i^{(1)} )^+ $ in all realizations.
The key point in the proof is that at each time basket $i$ was advanced a stage, its Weitzman index was larger than all other baskets', so its minimum index exceeds the minimum of the others, i.e. $\kappa_i^{(1)} \geq \kappa_j^{(1)}$ for all $j$.
Now, the algorithm is indeed a \emph{descending procedure}, so it is non-exposed by Lemma \ref{lemma:descending-nonexposed}.
By Lemma \ref{lemma:nested-key} then, \emph{(a)} its welfare is equal to its expected capped value $\E[\max_i (\kappa_i^{(1)})^+]$, and \emph{(b)} no algorithm can do better.

	\section{Approximation for Pandora's Matching Problem}\label{sec:unordered}

We now return to the matching setting, ready to utilize nested boxes.
In order to obtain an approximation algorithm for Pandora's Matching Problem, we consider \emph{oriented graphs} where inspections of each edge must proceed in a certain order.
We use this algorithm to create approximation algorithms to the general problem, even with some correlation between values.
Finally, we consider a tempting simplification to ``edge-based'' algorithms.

\subsection{Matching with Oriented Graphs}\label{subsec:matching-ordered}

We first suppose that each edge is restricted to a particular order of inspection, e.g. $(i,j)$ must be inspected prior to $(j,i)$.
This is a natural setting for some applications.
A job-search platform may require candidates to submit resumes before companies can screen them: inspections must begin on the candidate side before moving to the employer side of the match.
Formally, we capture the predetermined order of inspection with an orientation of the graph.

\begin{definition}\label{def:orientation}
	An \emph{orientation} $\O$ of an undirected graph $G = (V, E)$ is a set of ordered pairs of vertices such that, for every edge $\{i,j\} \in E$, $\O$ contains exactly one of $(i,j)$ or $(j,i)$.
\end{definition}

We refer to an ordered pair $(i,j)\in \O$ as an \emph{oriented edge}.
Given an instance and an orientation $\O$ of the edges, we say an algorithm is \emph{$\O$-oriented} if for all realizations its sequence of inspections respects the orientation, i.e. if $(i,j) \in \O$, then endpoint $(i,j)$ must be inspected prior to $(j,i)$, if $(j,i)$ is inspected.
If an algorithm is $\O$-oriented for some $\O$, we call it a \emph{fixed-orientation algorithm}.

\paragraph{Comparison to committing policies}
Fixed-orientation algorithms may be reminiscent of \emph{committing policies} appearing in related work on Pandora's Box with Non-obligatory Inspection, e.g. \citet{beyhaghi2019pandora}.
Committing policies fix the sequence of box inspections in advance, though they may choose their stopping time dynamically.
Although a fixed-orientation algorithm ``commits'' to an orientation $\O$, it will in general adaptively decide which edge or ``basket'' to advance next, depending on information revealed at previous stages.
In other words, it only commits to order of inspection \emph{within} each edge.
Another difference is that, in the nonobligatory setting, we can reduce WLOG to a set of $n+1$ commitments with $n$ boxes; it is not obvious how to reduce the $2^{|E|}$ orientations of a Pandora's Matching Problem instance.

\paragraph{Orientations and nested Pandora boxes}
For fixed-orientation algorithms, edges can be represented as Pandora baskets with two nested boxes, i.e. two stages of inspection.
Opening the first box reveals the value of one endpoint, and the opening second reveals the full value of the edge.
For example, if $(i,j) \in \O$, then edge $\{i,j\}$ is a basket where the first box costs $c_{ij}$ to open and reveals signal $v_{ij}$, and opening the second box for cost $c_{ji}$ reveals the combined value $v_{ij} + v_{ji}$.
To advance to the third stage, or ``claim the basket'', is to add the edge to our matching.

We now consider optimizing expected welfare when restricted to a fixed orientation $\O$.
Even under the fixed-orientation restriction, the optimal matching algorithm is unclear.
The problem is more general than matching with a single Pandora box on each edge, a problem suspected to be NP-hard.
But for a given orientation $\O$, the natural descending procedure achieves a $1/2$-approximation to the optimal $\O$-oriented algorithm.

\begin{algdef}[$\O$-Oriented Descending Procedure]\label{alg:ordered-inspection}
	Given a fixed orientation $\O$, for steps $t=1,\dots$:
	\begin{itemize}
		\item For each $(i,j) \in \O$, let $\sigma_{ij}^t$ denote the current Weitzman index of basket $(i,j)$.
		\item Advance the basket $(i,j)$ with largest $\sigma_{ij}^t$.
		\begin{itemize}
			\item If advancing results in ``claiming'' the basket, then add edge $\{i,j\}$ to the matching. Delete any edges incident to $i$ or $j$ from $\O$.
		\end{itemize}
		\item If at any point the largest $\sigma_{ij}^t$ is negative or $\O$ is empty, stop.
	\end{itemize}
\end{algdef}

This is indeed a \emph{descending procedure} (Definition \ref{def:descending-nonexposed}), as it maintains a nonincreasing eligible set (technically, $\O$ union the matched edges) and always advances and retains the largest-index basket in that set.
Lemma \ref{lemma:descending-nonexposed} therefore gives the following.
\begin{corollary}\label{cor:ordered-non-exposed}
	For any orientation $\O$, the $\O$-Oriented Descending Procedure is non-exposed.
\end{corollary}

We can continue with an analysis analogous to that of \citet{kleinberg2016descending} for the case of matching with a Pandora box on each edge.
Given $\O$, let \emph{$\O$-$\kappa$-Greedy} refer to the following matching (used only for analysis): for each oriented edge $(i,j) \in \O$, let $x_{\{i,j\}} = \kappa_{ij}^{(1)}$, the capped value for basket $(i,j)$; take the matching produced by the weighted greedy algorithm with edge weights $x_{\{i,j\}}$.

\begin{lemma}\label{lemma:ordered-inspection-greedy}
	For any orientation $\O$, with probability one, the $\O$-Oriented Descending Procedure \ref{alg:ordered-inspection} produces the same matching as $\O$-$\kappa$-Greedy.
\end{lemma}

\begin{proof}
	Fix a realization of all random variables.
	Suppose edge $e = \{i,j\}$ is matched by the $\O$-Oriented Descending Procedure.
	We abuse terms by referring to $e$ as a basket, where it is oriented according to $\O$.
	We show that $\kappa_e^{(1)}\geq \kappa_{e'}^{(1)}$ for any edge $e'$ that was available (i.e. legal to match) when $e$ is matched.

	Since basket $e$ has been fully advanced, all of its Weitzman indices $\sigma_e^{(1)},\sigma_e^{(2)},\sigma_e^{(3)}=v_e$ were at some point the largest current Weitzman index among all available baskets.
	In particular, since $e'$ is available when $e$ is matched, each index $\sigma_e^{(k)}$ is larger than $\sigma_{e'}^{(\ell)}$ for some $\ell$, so $\kappa_e^{(1)} = \min_k \sigma_e^{(k)} \geq \min_{\ell} \sigma_{e'}^{(\ell)} = \kappa_{e'}^{(1)}$.
\end{proof}

Using the characterization of the  $\O$-Oriented Descending Procedure's matching as the greedy matching, we show that the algorithm achieves a constant factor of the optimal welfare.

\begin{proposition} \label{prop:oriented-desc-half}
	The $\O$-Oriented Descending Procedure (Algorithm \ref{alg:ordered-inspection}) achieves at least $1/2$ the expected welfare of the optimal $\O$-oriented algorithm.
\end{proposition}
\begin{proof}
	Let $\A_{ij}^{\Opt}$ be the indicator that edge $\{i,j\}$ is claimed by the optimal $\O$-oriented algorithm.
	For any realization of all values, consider a graph with edge weights $\{\kappa_{ij}^{(1)}\}$.
	In this graph, let $\A_{ij}^G$ be the indicator for $\{i,j\}$ being included in the greedy matching, i.e. in $\O$-$\kappa$-Greedy; and let $\A_{ij}^*$ be the corresponding indicator for the maximum weighted matching.
	The welfare of the $\O$-Oriented Descending Procedure satisfies
	\begin{align*}
		\Welf(\Alg)
		&= \E \sum_{(i,j) \in \O} \A_{ij} \kappa_{ij}^{(1)}		& \text{Corollary \ref{cor:ordered-non-exposed}}  \\
		&= \E \sum_{(i,j) \in \O} \A_{ij}^G \kappa_{ij}^{(1)}		& \text{Lemma \ref{lemma:ordered-inspection-greedy}}  \\
		&\geq \frac{1}{2} \E \sum_{(i,j) \in \O} \A_{ij}^* \kappa_{ij}^{(1)}		& \text{weighted matching fact}  \\
		&\geq \frac{1}{2} \E \sum_{(i,j) \in \O} \A_{ij}^{\Opt{}} \kappa_{ij}^{(1)}	& \text{optimality of $\{\A_{ij}^* \}$}   \\
		&\geq \frac{1}{2} \Welf(\Opt) 				& \text{Lemma \ref{lemma:nested-key}} .
	\end{align*}
\end{proof}

\paragraph{Remark}
Even with a fixed orientation, we do not know how to improve on a $1/2$-approximation factor in polynomial time.
We are limited by reliance on greedy matchings, which can only achieve $1/2$ of the optimal weighted matching even without inspections.

\subsection{Main Positive Result}\label{subsec:matching-unordered}

We now finally provide two approximation algorithms for Pandora's Matching Problem, both of which first pick an orientation $\O$ and then run the $\O$-Oriented Descending Procedure.

The first algorithm we consider is randomized, and orients each edge independently.
We will first present and analyze the randomized algorithm, before presenting a deterministic algorithm using concepts developed in the randomized setting.

\begin{algdef}[Randomized Pandora Matching Algorithm]\label{alg:randomized}
	Construct an orientation $\O$ by orienting each edge $\{i,j\}$ as either $(i,j)$ or $(j,i)$ uniformly and independently at random.
	Run the $\O$-Oriented Descending procedure.
\end{algdef}

Note that by construction, $\O$ is drawn from the uniform distribution over all orientations for the graph.

\begin{theorem}\label{theorem:rand-approx}
	In Pandora's Matching Problem, Algorithm \ref{alg:randomized} achieves at least a $1/4$-approximation of the optimal welfare.
\end{theorem}

Before proving the approximation guarantee, we first define some helpful notation and an upper-bound on the value of an optimal algorithm.
We will also use these tools in defining our deterministic algorithm.

\begin{definition}\label{def:reverse-orientation}
	The \emph{reverse} of an orientation $\O$ is $\hat\O = \{(j,i)~|~ (i,j)\in \O\}$.
\end{definition}

The key step in the analysis of the randomized algorithm is an upper bound on the welfare of the optimal algorithm relative to an orientation and its reverse.
The crux of this bound is that any realization of an algorithm induces an orientation on edges it inspects, based on the order of inspection.
The welfare an edge generates in an algorithm, then, can be ``charged'' across any orientation $\O$ and its reverse, according to which orientation matches the direction in which the edge is inspected.

\begin{lemma}\label{lemma:opt-bound}
	For any orientation $\O$ and it s reverse $\hat{\O}$, the welfare of \Opt~on Pandora's Matching Problem is upper-bounded by
	\begin{align*}
		\Welf(\Opt) &\leq  2\cdot \Welf(\text{$\O$-desc.}) ~~+~ 2\cdot \Welf(\text{$\hat{\O}$-desc.}).
	\end{align*}
\end{lemma}

The proof of this lemma is given in Appendix \ref{sec:unordered-appdx}.
This bound on the optimal welfare of any algorithm for Pandora's Matching Problem allows a quick proof of Theorem \ref{theorem:rand-approx}, also in Appendix \ref{sec:unordered-appdx}.

Given the heavy reliance of Algorithm \ref{alg:randomized} on randomization of orientation, we also present a deterministic algorithm based on determining an orientation and then running the associated descending procedure.
A similar “best-of-two-worlds” technique is used in \citet{beyhaghi2019pandora} in a different context, and is considered in Appendix \ref{sec:unordered-appdx}.

This algorithm differs from the randomized algorithm in that the final orientation is determined globally.
In the randomized algorithm we select the orientation of each edge independently, but the Best-of-Two-Worlds algorithm compares two orientations over the entire graph and selects one.

\begin{theorem}\label{theorem:best-of-two-approx}
	In Pandora's Matching Problem, the Best-of-Two-Worlds algorithm achieves at least a $1/4$-approximation of the optimal welfare.
\end{theorem}

\paragraph{Extension to correlated-within-edges model}
We can extend these positive results to a model where each pair of values $v_{ij}, v_{ji}$ on an edge are arbitrarily correlated with each other and independent of the other edges.
In fact, we can extend this a bit further: in the \emph{correlated-within-edges model}, we suppose that each Pandora box on an endpoint $(i,j)$ reveals for cost $c_{ij}$ an arbitrary signal $s_{ij}$, the total value on an edge is an arbitrary measurable function $f_{\{i,j\}}(s_{ij}, s_{ji})$, and for each edge, the pair $(s_{ij}, s_{ji})$ are drawn jointly from a distribution $D_{\{i,j\}}$, independently of all other pairs.

In this setting, the definition of an algorithm is unchanged, i.e. it adaptively pays to open boxes and eventually outputs a matching of fully-opened edges.
The proofs of Proposition \ref{prop:oriented-desc-half} and Lemma \ref{lemma:opt-bound} go through unchanged, and the proofs of both algorithms hold with values correlated along edges.
In the proof of Lemma \ref{lemma:opt-bound}, an edge is modeled as a pair of arbitrary Pandora baskets $(i,j)$ and $(j,i)$, where an algorithm must pick one of the baskets at the time of first inspection of the edge.
For any orientation $\O$, i.e. a commitment on each edge to one of its two possible baskets, the proof then cites Proposition \ref{prop:oriented-desc-half}.
And Proposition \ref{prop:oriented-desc-half} addresses the problem of matching with a single, arbitrary Pandora basket per edge, utilizing our generic results for Nested Boxes in Section \ref{sec:multistage}.

\begin{observation}
	The $1/4$-approximations of the Randomized Pandora Matching and  Best-of-Two-Worlds algorithms continue to hold in the correlated-within-edges model.
\end{observation}

\subsection{Dessert?}\label{subsec:unordered-inspections}

We have finally presented two constant-factor approximation algorithms for Pandora's Matching Problem, one randomized and one deterministic.
The extreme simplicity of our randomized algorithm's edge orientation choice hints that a simpler deterministic approach should be possible, potentially following the structure of the randomized algorithm: for each edge, pick an orientation independently of any other edge orientations.
This would greatly simplify our deterministic approximation algorithm, and such an algorithm would be a nice ``dessert''.
In particular, a natural approach might be to orient each edge according to which ``basket'' has a higher initial Weitzman index, then run the Oriented Descending Procedure.
We consider the following generalization of this approach:

\begin{definition} \label{def:edge-based}
	An \emph{edge-based fixed-orientation algorithm} is one that, given an instance of Pandora's Matching Problem, constructs an orientation $\O$ by applying a deterministic rule to each edge $\{i,j\}$ that depends only on the parameters of that edge, $D_{ij},c_{ij},D_{ji},c_{ji}$, and then runs some $\O$-oriented algorithm.
\end{definition}

Unfortunately, we have the following result.
\begin{theorem}[``No dessert''] \label{thm:no-dessert}
	No edge-based fixed-orientation algorithm guarantees any positive approximation for Pandora's Matching Problem.
\end{theorem}

\begin{figure}
	\begin{center}
		\begin{tikzpicture}[
	every node/.style={circle,draw},scale=1]

	\node (i) at (-3,0) {$i$};
	\node (j) at (2,0) {$j$};
	\draw (i) node[right,rectangle,above,sloped,draw=none,fill=none,outer ysep=8pt]{} -- (j) node[left,rectangle,above ,sloped,draw=none,fill=none,outer ysep=8pt]{};

	\node[right,rectangle,below,sloped,draw=none,fill=none,outer ysep=8pt,align=left] (vij) at ($(i)$) {$v_{ij} = \begin{cases}
			\frac{1}{\alpha^3} & \text{w.prob. $\alpha$}  \\
			0  & \text{else}
		\end{cases}$  \\[0.5em]
	 	$c_{ij} = 1$};
	\node[left,rectangle,below,sloped,draw=none,fill=none,outer ysep=8pt, align=left] (vji) at ($(j)$) {$v_{ij} = \begin{cases}
			 0  & \text{w.prob. $\alpha^2$}  \\
			-\frac{1}{\alpha^3}+\frac{1}{\alpha}  & \text{else}
		\end{cases}$  \\[0.5em]
		$c_{ji} = 1 - \alpha$};
\end{tikzpicture}
	\end{center}
	\caption{An edge used to prove Theorem \ref{thm:no-dessert}.
		Without an outside option, it is optimal to inspect $i$ before $j$ (Instance 1), but with an outside option of $\frac{1}{\alpha}$, it is optimal to inspect $j$ before $i$ (Instance 2).
	}
	\label{fig:no-dessert}
\end{figure}
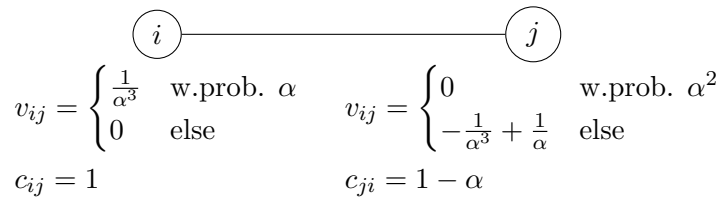

\paragraph{Intuition}
The proof relies on a particular edge, shown in Figure \ref{fig:no-dessert}.
Every edge-based fixed orientation algorithm either orients that edge in one direction or the other.
If there is no outside option, then the $(i,j)$ orientation is arbitrarily better, because we want to match the edge if and only if $v_{ij} > 0$, and we discover this immediately, saving the cost of inspecting $(j,i)$ when $v_{ij} = 0$.
But if there is an outside option of $\frac{1}{\alpha}$, then the basket is worthwhile if and only if both boxes contain their high value (a ``success'').
Here the $(j,i)$ orientation is arbitrarily better, as the overall probability of a success is the same but we pay a second inspection cost with only $\alpha^2$ probability instead of $\alpha$.

For the skeptical reader, we also give an alternative proof of Theorem \ref{thm:no-dessert} in Section \ref{subapp:no-dessert-alternate} that does not use any Pandora box machinery.

\begin{proof}
	Recall that $\sigma_{ij}^{(1)}$ is the initial Weitzman index for the basket corresponding to orientation $(i,j)$, and $\kappa_{ij}^{(1)}$ is the capped value, a random variable.
	By Lemma \ref{lemma:nested-key}, the expected welfare contribution of a basket $(i,j)$ is upper-bounded $\E[\A_{ij} \kappa_{ij}^{(1)}]$, with equality if the algorithm is non-exposed.
	In particular, since $\kappa_{ij}^{(1)} \leq \sigma_{ij}^{(1)}$, we can upper-bound the welfare contribution of a box under any algorithm by the maximum $\sigma_{ij}^{(1)}$ over any basket available to it.

	We record the following (calculations appear in Appendix \ref{subapp:no-dessert}):
	\begin{itemize}
		\item $\sigma_{ij}^{(1)} = \frac{1}{\alpha} - 1$.
		\item $\kappa_{ij}^{(1)} = \frac{1}{\alpha} - 1$ with probability $\alpha$, otherwise negative.
		\item $\sigma_{ji}^{(1)} = \frac{1}{\alpha^2} - \frac{1}{\alpha}$.
		\item $\kappa_{ji}^{(1)} = \frac{1}{\alpha^2} - \frac{1}{\alpha}$ with probability $\alpha^3$, otherwise nonpositive.
	\end{itemize}

	\textbf{Instance 1.}
	We consider a graph consisting of just the edge in Figure \ref{fig:no-dessert}.
	An algorithm can orient the edge as a basket $(i,j)$ or as a basket $(j,i)$.
	In each case, the descending procedure is optimal and achieves welfare equal to expected capped value (Section \ref{sec:multistage}).
	With a single $(i,j)$ basket, the descending procedure obtains welfare $\E[(\kappa_{ij})^+] = \alpha \left(\frac{1}{\alpha} - 1\right) = 1 - \alpha$.
	With $(j,i)$, the welfare is $\E[(\kappa_{ji}^{(1)})^+] = \alpha^3 \left(\frac{1}{\alpha^2} - \frac{1}{\alpha}\right) = \alpha - \alpha^2$.
	So the orientation $(j,i)$ first is an approximation of $\frac{\alpha - \alpha^2}{1-\alpha} = \alpha$.
	This holds for arbitrarily small $\alpha > 0$, so no algorithm that orients the edge as $(j,i)$ can guarantee a positive approximation.

	\vskip1em
	\textbf{Instance 2.}
	We will consider a star graph with $i$ at the center.
	There is a special edge $\{i,k\}$ with $v_{ik} = \frac{1}{\alpha}$, $c_{ik} = 0$, $v_{ki} = 0$, $c_{ki} = 0$.
	In other words, $i$ has an ``outside option'' $\frac{1}{\alpha}$ it can match to at any time.
	Every other edge in the graph $\{i,j_1\},\dots,\{i,j_m\}$ is an independent copy of the edge $\{i,j\}$ in Figure \ref{fig:no-dessert}.
	The number of copies will be $m = \omega(\tfrac{1}{\alpha^3})$.
	A edge-based fixed-orientation algorithm must orient each copy in the same direction, as it applies the same deterministic rule to each edge.

	After orienting every edge in the $(j,i)$ direction, we have an instance of Nested Pandora's Boxes with $m+1$ baskets.
	The optimal welfare, achieved by the descending procedure, is equal to the expected maximum capped value of a basket.
	With $\omega(\tfrac{1}{\alpha^3})$ baskets, some basket has capped value $\frac{1}{\alpha^2} - \frac{1}{\alpha}$ with probability $1 - o(1)$; so the algorithm's expected welfare is $(\frac{1}{\alpha^2} - \frac{1}{\alpha})\left(1 - o(1)\right)$.

	After orienting every edge in the $(i,j)$ direction, the same argument holds, except that the capped value of every basket is always at most $\frac{1}{\alpha} - 1$.
	These are actually less than the capped value of the outside option, which is deterministically $\frac{1}{\alpha}$, which is the maximum capped value and therefore the best any algorithm can do.

	The ratio of the optimal $(i,j)$ orientation welfare to $(j,i)$ is, asymptotically, $\frac{1/\alpha}{1/\alpha^2} = \alpha$.
	This holds as $\alpha \to 0$, so no algorithm that orients $\{i,j\}$ as $(i,j)$ has a positive approximation guarantee.
\end{proof}

The analysis of instance 1 provides an interesting corollary.

\begin{corollary}
	Even with just a single edge, inspecting according to the orientation $(i,j)$ with the larger Weitzman index $\sigma_{ij}^{(1)}$ can be arbitrarily suboptimal.
\end{corollary}

	\section{Discussion}

\paragraph{Summary}
Prior work had considered matching with a single Pandora box on each edge.
A satisfactory solution exists for that problem with by-now standard techniques for information-acquisition settings.
In this paper, we considered Pandora's Matching Problem, a two-sided matching setting in which each participant $i$ and $j$ in a potential match $\{i,j\}$ have an initially-hidden value for the match, modeled as one Pandora box on each endpoint of an edge.
We allowed for negative values, an important case for settings such as labor markets.

We found that this problem poses significant challenges for standard approaches in ways that likely have implications for the design of matching platforms.
In particular, the following approaches all fail to guarantee any positive approximation factor.
\begin{enumerate}
	\item A \emph{bundled-box} approach, where inspections of both boxes incident to an edge must occur simultaneously.
		This suggests that interview stages on matching platforms can be costly in terms of welfare, while asynchronous inspections should be enabled.
	\item A \emph{vertex-based} approach, where decisions are made based on values and Weitzman indices at the endpoints, but the algorithm does not consider the interaction of the value distributions of the endpoints.
		(This approach gives an approximation factor in the positive-values setting.)
		This suggests that platforms should treat an edge not as two separate boxes, but as one basket of \emph{nested boxes} that undergoes a multi-stage information-gathering process.
	\item An \emph{edge-based} approach, where the order of inspection within an edge is prescribed by some deterministic rule independent of the rest of the graph.
		This suggests that platforms should coordinate the order of inspections within edges (i.e $i$ first or $j$ first) in a global fashion.
\end{enumerate}

To evade all of these pitfalls, we proposed a Best-of-Two-Worlds algorithm that: (1) dynamically and partially inspects edges, only coming back to complete an edge's inspection if prudent; (2) treats each edge as a multi-stage Pandora basket with nested boxes; (3) coordinates the decision of how to orient each edge globally by considering the ``two worlds''.
We showed that this algorithm gives an efficient $1/4$-approximation for Pandora's Matching Problem.

\subsection{Future Directions}

A number of questions remain for two-sided matching with inspection costs.
One interesting question is in the difficulty of the Pandora's Matching Problem itself.
Is there an algorithm for the Pandora's Matching Problem that beats the $1/2$-approximation factor from greedy matching?
\begin{openprob}
	Is Pandora's Matching Problem \textsf{\APX}-hard, or can a PTAS be found?
\end{openprob}
It is also open to give a substantially different approximation algorithm than Best-of-Two-Worlds.

\paragraph{Extensions}
There are other extensions of the Pandora's Box model which may also be interesting for modeling two-sided matching settings.
The \emph{nonobligatory} setting allows boxes to be claimed uninspected.
If, for example, a student knows she likely has a high value for attending a certain school, she may forgo the cost of visiting the school and directly accept an offer.
The simpler \emph{Pandora's Box with Nonobligatory Inspection} problem without the matching constraint is already known to be \NP-hard \citep{fu2023pandora}, but existing approximation schemes \citep{fu2023pandora, beyhaghi2023pandora} give hope that approximation may be tractable in the matching setting as well.

Section \ref{sec:unordered} showed that edge-based fixed-orientation algorithms are, in general, arbitrarily suboptimal.
However, it is possible that under natural restrictions on the value distributions, such algorithms do achieve constant-factor approximations.
Understanding settings in which such algorithms can safely be applied could have implications for design of platforms which fix an orientation for practicality.

	\subsection*{Acknowledgements}
	This work was supported by the National Science Foundation Award \#2329431.

	\bibliographystyle{ACM-Reference-Format}
	\bibliography{citations}

	\appendix

	\section{Vertex-Based Algorithms \& Proofs}\label{sec:separate-models-appdx}

This section contains a description of the vertex-based algorithm from Section \ref{subsec:vertex-based} and its analysis, as well as further omitted proofs from Section \ref{subsec:vertex-based}.

\subsection{1/4-Approximation for Positive-Values}\label{subsec:positive-values-appdx}

We define our simplification of the mechanism given in \citet{bowers2023high} and \citet{waggoner2019matching} here, and prove it obtains a $1/4$-approximation in the positive-values setting.

\begin{algdef}[Vertex-Based Descending Procedure]\label{alg:quarter-approx}
	While there is a feasible edge that is not matched, consider the highest $\sigma_{ij}$ of unopened boxes and the highest $v_{k\ell}$ of opened boxes.
	If $\sigma_{ij} > v_{k\ell}$, open box $(i,j)$.
	Otherwise, match the edge $\{k,\ell\}$, opening the $(\ell, k)$ endpoint if necessary.
\end{algdef}

This procedure is similar to the descending algorithm given by \citet{kleinberg2016descending}: repeatedly open or claim the highest-available box.
It is also vertex-based, according to Definition \ref{def:vertex-based}.

We will show that the vertex-based descending procedure is a $1/4$-approximation for the positive-valued setting, but first we prove several helpful intermediate facts.

\begin{observation}\label{obs:pos-approx-non-exposed}
	The Vertex-Based Descending Procedure \ref{alg:quarter-approx} is non-exposed on every ordered pair $(i,j)$.
\end{observation}

\begin{proof}
	If the algorithm inspects $(i,j)$ and finds that $v_{ij}\geq \sigma_{ij}$, the pair is added back at the top of the list as $\sigma_{ij}$ was previously the largest entry in the list.
	At the next step, the algorithm will examine the same pair again, find that it has already been inspected, and add it to the matching.
\end{proof}

We compare the welfare of this algorithm to the greedy matching algorithm on edges with weights $w_{\{i,j\}} = \max(\kappa_{ij}, \kappa_{ji})$.
We call this the max-$\kappa$-Greedy matching and denote it by the random variable $\A^G$.
It is unclear whether there exists a box-opening algorithm which achieves this matching, since the greedy algorithm relies on revealed capped values $\kappa$ which are not initially accessible to an algorithm.
We show that the vertex-based matching algorithm, in fact, produces a matching which is exactly the greedy-by-max-kappa matching.

\begin{claim}
	With probability one, the vertex-based descending procedure \ref{alg:quarter-approx} produces the same matching as max-$\kappa$-Greedy for all realizations of values.
\end{claim}

\begin{proof}
	Fix some realization of values, and suppose the vertex-based algorithm \ref{alg:quarter-approx} matches the edge $\{i, j\}$.
	Without loss of generality, assume the edge was matched from the $(i,j)$ endpoint, so $\kappa_{ij} = \max(\kappa_{ij}, \kappa_{ji})$.

	We show that $\kappa_{ij}\geq \kappa_{k\ell'}$ for all other feasible pairs $(k, \ell)$.

	Assume by way of contradiction that there exists another ordered pair $(k, \ell)$ such that $\kappa_{ij} < \kappa_{k \ell}$.
	Then either $\sigma_{ij} < \kappa_{k \ell}$ or $v_{ij} < \kappa_{k\ell}$.

	If $\sigma_{ij} < \kappa_{k \ell}$, then $\sigma_{ij} < \sigma_{k \ell}$ and the algorithm would have inspected $(k, \ell)$ before inspecting $(i, j)$.
	After inspecting, the algorithm would find $v_{k\ell} > \sigma_{ij}$ and would match $\{k, \ell\}$ before matching $\{i, j\}$, a contradiction.

	Similarly, if $v_{ij} < \kappa_{k\ell}$, then $v_{ij}< \sigma_{k \ell}$ and the algorithm would have inspect $(k, \ell)$ before matching $\{i, j\}$.
	After inspecting, the algorithm would reinsert $(k, \ell)$ into the list with value $v_{k\ell}>v_{ij}$, and would match $\{k, \ell\}$ before matching $\{i, j\}$, also a contradiction.

	The algorithm adds the feasible edge endpoint $\{i,j\}$ with the largest $\kappa_{ij}$ as the next edge to the matching.
	Since and edge is listed twice by $(i, j)$ and $(j, i)$, the next added edge then maximizes $\max(\kappa_{ij}, \kappa_{ji})$, and the matching created is exactly max-$\kappa$-Greedy.
\end{proof}

Finally, we use these facts about the Vertex-Based Descending Price algorithm to prove our 1/4-approximation result.

\begin{proof}[Proposition \ref{prop:pos-approx}]
	We denote the optimal matching produced by any box-opening algorithm $\A^{\Opt}$.
	By Lemma \ref{lemma:KWW-key}, the value of the optimal policy is $\Opt \leq \E[\sum_{(i,j)\in E}\A^\Opt_{ij}(\kappa_{ij} + \kappa_{ji})]$.
	Since the algorithm is non-exposed with positive valuations, the expected value of the algorithm is
	\begin{align*}
		\E\left[\sum_{(i,j)}\A_{ij}v_{ij} - \I_{ij}c_{ij}\right] &= \E\left[\sum_{\{i,j\}\in E}\A_{ij}(\kappa_{ij} + \kappa_{ji})\right]\\
		&\geq \E\left[\sum_{\{i,j\}\in E}\A_{ij}\max(\kappa_{ij}, \kappa_{ji})\right]\\
		&= \E\left[\sum_{\{i,j\}\in E}\A^\text{G}_{ij}\max(\kappa_{ij}, \kappa_{ji})\right],
	\end{align*}
	where $\A^{G}_{ij}$ is the indicator for max-$\kappa$-Greedy.
	Since a greedy matching algorithm is a 1/2-approximation of the optimal matching on the given edge weight, this is at least half the value of the optimal matching $\A^*$ on max-kappa-weighted edges.
	The welfare optimal matching on max-$\kappa$-weighted edges is, in turn, an upper bound on the welfare of the optimal matching on edges with different weights, particularly the weights $w_{\{i,j\}} = \kappa_{ij}+ \kappa_{ji}$.
	Combining these steps mathematically then,
	\begin{align*}
		\E\left[\sum_{(i,j)\in E}\A^\text{G}_{ij}\max(\kappa_{ij}, \kappa_{ji})\right] &\geq \frac{1}{2} \E\left[\sum_{(i,j)\in E}\A^*_{ij}\max(\kappa_{ij}, \kappa_{ji})\right]  \\
		&\geq \frac{1}{2} \E\left[\sum_{(i,j)\in E}\A^\Opt_{ij}\max(\kappa_{ij}, \kappa_{ji})\right]\\
		&\geq \frac{1}{4} \E\left[\sum_{(i,j)\in E}\A^\Opt_{ij}(\kappa_{ij} +  \kappa_{ji})\right] \geq \frac{\Opt}{4}.
	\end{align*}
\end{proof}

This reliance on the max-endpoint greedy matching causes the Vertex-Oriented Descending Procedure (Definition \ref{alg:quarter-approx}) to immediately fail on instances where values may be negative.

\subsection{Analysis of Indistinguishable Edge}\label{subsec:indistinguishable-appdx}

We use a single edge for which the endpoints have a unique optimal sequence of inspection, but which cannot be distinguished by an edge-based algorithm.

The idea will be that the edge's endpoints are indistinguishable to vertex-based algorithms: they have the same cost of inspection and Weitzman index.
However, to achieve nontrivial welfare, it is crucial to inspect the endpoints in the correct order based on their differing distributions of values.

\begin{example}[Indistinguishable Edge ]\label{ex:vertex-based-counterexample}
	The graph consists of a single edge $\{i,j\}$.
	The Pandora box $(j,i)$ has deterministic value $v_{ji} = 2$ and cost $c_{ji} = 1$.
	The Pandora box $(i,j)$ has value $v_{ij} = 9$ with probability $1/8$ and $v_{ij} = -3$ otherwise, and has cost $c_{ij} = 1$.
\end{example}

We analyze Example \ref{ex:vertex-based-counterexample} here, and find that an optimal algorithm can achieve a positive value by correctly ordering the inspection of endpoints, while a vertex-based algorithm cannot distinguish the endpoints to choose the correct order.

Similar to the proof of Proposition \ref{prop:simul-fails}, the key is to observe that one direction of inspection obtains information immediately while the other pays an extra cost for unhelpful information.
Randomizing the inspection order is not enough to obtain a constant approximation because one direction has negative welfare.

Faced with the possibility of inspecting the wrong order, any inspecting vertex-based algorithm expects to receive a negative welfare and would rather do nothing.
This will hold even if the vertex-based algorithm is randomized.
We break this proof down into an analysis of the optimal solution, the claim that a vertex-based algorithm cannot distinguish the endpoints, and combine the two for a proof of Proposition \ref{prop:vertex-based-fails}.

\begin{claim}\label{claim:indistinguishable-optimal-utility}
	\Opt{} on Example \ref{ex:vertex-based-counterexample} achieves welfare $1/4$.
\end{claim}

\begin{proof}
	We consider the options available to an algorithm: starting its inspection with vertex $i$ and with $j$.
	(The algorithm can also do nothing, for welfare zero.)

	If it inspects vertex $i$ first, then it either finds $v_{ij} = 9$ or $v_{ij} = -3$.
	If $v_{ij} = -3$, it should terminate immediately as the value from $v_{ji}$ will not cover the loss from $v_{ij}$.
	If $v_{ij} = 9$, however, it should inspect $(j,i)$ and claim the edge for a value of $11$.
	Thus by inspecting vertex $i$ first, and then inspecting and matching $j$ if and only if $v_{ij} = 9$, the algorithm gets expected value $\E[\A(v_{ij} + v_{ji} - c_{ji})] - c_{ij} = (9 + 2 - 1)/8 - 1 = 1/4$.

	If the algorithm instead inspects $j$ first, it gains no information about $i$ for a cost of $1$.
	If it chooses not to inspect $i$, then its net welfare is $-1$.
	On the other hand, if it inspects $i$, it has a $1/8$ chance of finding a high value worth claiming the edge.
	This gives an overall value of $\E[\A(v_{ij} + v_{ji})] - c_{ji} - c_{ij} = -5/8$.
	While better than not performing the second inspection, it is still better to perform no inspections at all than to start from $(j,i)$.
\end{proof}

\begin{lemma}\label{lemma:vertex-based-bad}
	Prior to inspection, any vertex-based algorithm cannot distinguish the endpoints of the edge in Example \ref{ex:vertex-based-counterexample}.
\end{lemma}

\begin{proof}
	It is sufficient to show that the values that the algorithm has access to prior to inspection ($c_{ij}$, $\sigma_{ij}$) are the same for both endpoints.
	The costs $c_{ij},c_{ji}$ are both deterministically $1$ and therefore indistinguishable, so we compute the indices $\sigma_{ij}$.

	For $(j,i)$, $\sigma_{ji}$ is the solution to $1 = \E[(v_{ji} - \sigma_{ji})^+] = (2 - \sigma_{ji})^+$, so $\sigma_{ji} = 1$.
	For the endpoint $(i,j)$, we solve $1 = \E[(v_{ij} - \sigma_{ij})^+] = \tfrac{1}{8}(9 - \sigma_{ij})^+$, so $\sigma_{ij} = 1$ as well.
\end{proof}

\begin{proof}[Proposition \ref{prop:vertex-based-fails}]
	Given an algorithm, we present it with Example \ref{ex:vertex-based-counterexample}, uniformly randomizing the order of the endpoints (i.e. the identities of $i$ and $j$).

	Any algorithm which performs no inspections gets value $0$.

	Otherwise, since the algorithm cannot distinguish the endpoints due to our randomization, it has a $1/2$ probability of beginning its inspection with vertex $i$ and $1/2$ with $j$.
	The best any algorithm can do, if it does inspect, is the average of the best outcomes if it began inspection from $i$ or from $j$:
	\begin{align*}
		\frac{1}{2}\left(\E[\A(v_{ij} + v_{ji} - c_{ji})] - c_{ij}\right) + \frac{1}{2}\left(\E[\A(v_{ij} + v_{ji})] - c_{ji} - c_{ij}\right) = \frac{1}{8} - \frac{5}{16} = -3/16.
	\end{align*}
	Therefore, the optimal vertex-based algorithm on this instance performs no inspections.
\end{proof}

	\section{Omitted Proofs for Nested Pandora's Box}\label{sec:nested-appdx}

This section provides the proofs omitted from Section \ref{sec:multistage}.

First, we prove a fact that will be helpul in the proof of the key lemma, Lemma \ref{lemma:nested-key}.
\begin{lemma}\label{lemma:utility-kappas}
	The expected welfare contribution of box $i$ under any algorithm satisfies
	\begin{equation*}
		\E\left[\A_iv_i - \sum_{\ell = 1}^d \I_i^{(\ell)}c_i^{(\ell)}\right]
		= \E\left[\A_iv_i - \sum_{\ell = 1}^d \I_i^{(\ell)}\left(\kappa_i^{(\ell+1)}-\kappa_i^{(\ell)}\right)\right] .
	\end{equation*}
\end{lemma}
Intuitively, this proof is relatively direct from the definition of nested Weitzman index,
  \[ \E\left[ \left(\kappa_i^{(\ell+1)} - \sigma_i^{(\ell)}\right)^+ ~\vline~ s_i^{(1)},\dots,s_i^{(\ell-1)} \right] = c_i^{\ell} , \]
along with the observation that $\I_i^{(\ell)}$ is independent of $(\kappa_i^{(\ell+1)} - \sigma_i^{(\ell)})^+$ conditioned on $s_i^{(1)},\dots,s_i^{(\ell-1)}$.
(This follows because the decision to open box $\ell$ comes prior to observing any of the signals $s_i^{(\ell)},\dots,s_i^{(d-1)}$ and value $v_i$ that define $\kappa_i^{(\ell+1)}$.)
However, the precise argument is subtle, so we go somewhat slowly and carefully.
\begin{proof}
	To simplify notation for this proof, let $s_{i,\ell} := (s_i^{(1)}, \dots, s_i^{(\ell-1)})$.

	We first expand the second term to explicitly represent the dependence on received signals.
	We use $\E_{s_{i,\ell}} \left[ \E[ \cdot \mid s_{i,\ell}] \right]$ to denote an outer expectation taken over the random draws of $s_i^{(1)},\dots,s_i^{(\ell-1)}$, and an inner expectation of $(\cdot)$ conditioned on $s_i^{(1)},\dots,s_i^{(\ell-1)}$ and taken over all other randomness.
	\begin{align*}
		\E\left[\sum_{\ell = 1}^d \I_i^{(\ell)}c_i^{(\ell)}\right]
		&= \sum_{\ell = 1}^d \E_{s_{i,\ell}} \left[ \E \left[ \I_i^{(\ell)} ~\vline~ s_{i,\ell} \right] c_i^{(\ell)} \right] .
	\end{align*}
	Now, we substitute in for $c_i^{(\ell)}$ using the definition of the nested Weitzman index and capped value.
	To be precise, it is important to note that the randomness in the definition of $c_i^{(\ell)}$ is over a different, independent space, whose first $\ell-1$ signals we will denote analogously by $s_{i,\ell}'$.
	\begin{align}
		\E_{s_{i,\ell}} \left[ \E \left[ \I_i^{(\ell)} ~\vline~ s_{i,\ell} \right] c_i^{(\ell)} \right]
		&= \E_{s_{i,\ell}} \left[ \E \left[ \I_i^{(\ell)} ~\vline~ s_{i,\ell} \right] \E_{s_{i,\ell}'} \left[ \E\left[ \left(\kappa_i^{(\ell+1)} - \sigma_i^{(\ell)}\right)^+ ~\vline~ s_{i,\ell}' \right] \right] \right]  \nonumber \\
		&= \E_{s_{i,\ell}} \E_{s_{i,\ell}'} \left[ \E \left[ \I_i^{(\ell)} ~\vline~ s_{i,\ell} \right] \E \left[ \left(\kappa_i^{(\ell+1)} - \sigma_i^{(\ell)}\right)^+ ~\vline~ s_{i,\ell}' \right] \right]  \nonumber \\
        &= \E_{s_{i,\ell}} \left[ \E \left[ \I_i^{(\ell)} ~\vline~ s_{i,\ell} \right] \E \left[ \left(\kappa_i^{(\ell+1)} - \sigma_i^{(\ell)}\right)^+ ~\vline~ s_{i,\ell} \right] \right]  \label{eqn:redraw} \\
        &= \E_{s_{i,\ell}} \left[ \E \left[ \I_i^{(\ell)} \left(\kappa_i^{(\ell+1)} - \sigma_i^{(\ell)}\right)^+ ~\vline~ s_{i,\ell} \right] \right]  \label{eqn:cond-indep} \\
        &=  \E \left[ \I_i^{(\ell)} \left(\kappa_i^{(\ell+1)} - \sigma_i^{(\ell)}\right)^+ \right] .  \nonumber
	\end{align}
	Equality \ref{eqn:redraw} follows because the worlds of $s_{i,\ell}$ and $s_{i,\ell}'$ have the same distribution.
	Equality \ref{eqn:cond-indep} follows by conditional independence, as discussed above.
	It only remains to observe that $(\kappa_i^{(\ell+1)} - \sigma_i^{(\ell)})^+ = \kappa_i^{(\ell+1)} - \min\{\kappa_i^{(\ell+1)}, \sigma_i^{(\ell)} \} = \kappa_i^{(\ell+1)} - \kappa_i^{(\ell)}$.
\end{proof}

We use this to prove the key lemma \ref{lemma:nested-key} relating actual utility to an amortized version, which is used throughout the paper.

\begin{proof}[Lemma \ref{lemma:nested-key}]
	First, Lemma \ref{lemma:utility-kappas} gives
	\begin{align*}
		\E\left[\A_iv_i - \sum_{\ell = 1}^d \I_i^{(\ell)}c_i^{(\ell)}\right] 	&= \E\left[\A_iv_i - \sum_{\ell = 1}^d \I_i^{(\ell)}\left(\kappa_i^{(\ell+1)}-\kappa_i^{(\ell)}\right)\right] .
	\end{align*}
	The proof is direct, but depends crucially on the fact that $\I_i^{(\ell)}$ and $\kappa_i^{(\ell+1)}$ are independent conditioned on $s_i^{(1)},\dots,s_i^{(\ell-1)}$.
	Let $\ell^*$ be a random variable denoting the last inspection stage the algorithm performs for basket $i$.
	If the algorithm performs no inspections on $i$, then $\ell^* = 0$, and if it performs all inspections (including the case where $\A_i = 1$), then $\ell^* = d$.
	We have
	\begin{align*}
		\E\left[\A_iv_i - \sum_{\ell = 1}^d \I_i^{(\ell)}c_i^{(\ell)}\right] &= \E\left[\A_iv_i - \sum_{\ell = 1}^{\ell^*} \left(\kappa_i^{(\ell+1)}-\kappa_i^{(\ell)}\right)\right]\\
		&=    \E\left[\A_i \kappa_i^{(d+1)} - \left(\kappa_i^{(\ell^*+1)} - \kappa_i^{(1)}\right)\right]  \\
		&\leq \E\left[\A_i\left(\kappa_i^{(d+1)} - \kappa_i^{(\ell^*+1)} + \kappa_i^{(1)}\right)\right]  \\
		&= \E\left[\A_i \kappa_i^{(1)} \right] ,
	\end{align*}
	where the inequality holds because $\A_i \geq 0$ and $\kappa_i^{(\ell)}$ is nondecreasing in $\ell$; and the final equality holds in both the case $\A_i = 0$ and the case $\A_i = 1$ (in which $\ell^* = d$).
	Furthermore, the inequality is strict if and only if, with nonzero probability, $\A_i = 0$ and $\kappa_i^{(\ell^* + 1)} > \kappa_i^{(1)}$ for some $\ell^* \in \{1,\dots,d\}$.
	Observe that $\kappa_i^{(1)} = \min\{ \gamma_i^{(\ell^*)}, \kappa_i^{(\ell^* + 1)} \}$, so this is equivalent to $\kappa_i^{(\ell^*+1)} > \gamma_i^{(\ell^*)}$ and $\I_i^{(\ell^*+1)} > \I_i^{(\ell^*)}$: the algorithm being exposed.
\end{proof}

\subsection{Descending Procedure Proofs}\label{subsec:descending-appdx}

We write $\ell_i(t) \geq 0$ to denote the number of opened boxes in basket $i$ at the beginning of step $t$.
For notational convenience, let $\sigma_i^t := \sigma_i^{(\ell_i(t)+1)}$, the index of the box we are \emph{considering} opening in basket $i$ at time $t$; Similarly, let $\kappa_i^t := \kappa_i^{(\ell_i(t)+1)}$ and $\gamma_i^t := \gamma_i^{(\ell_i(t)+1)}$.
We refer to $\sigma_i^t$ as the \emph{current Weitzman index} of basket $i$ at step $t$.

The proof of Lemma \ref{lemma:descending-nonexposed} is not as immediate as one might expect.
As mentioned above, there can be cases in which the index $\sigma_i^t$ decreases after advancing basket $i$, yet to avoid non-exposure, we must continue advancing.
A key observation is that the descending procedure always advances a basket with maximum $\gamma_i^t$.
However, this is not quite enough on its own either, as an increase in $\sigma_i^t$ does not correspond to an increase in $\gamma_i^t$.
The second key observation is that the descending procedure only switches away from a basket $i$ when $\gamma_i^t = \sigma_i^t$.
These points are captured in the following lemma.

\begin{lemma}\label{lemma:gamma-equality}
	At each time $t$, a descending procedure satisfies the following: it always advances an eligible basket $i$ with maximum $\gamma_i^t$, and for all other eligible $j\neq i$, $\gamma_j^t = \sigma_j^t$.
\end{lemma}

\begin{proof}[Lemma \ref{lemma:gamma-equality}]
	We argue by induction on the step of the algorithm.
	At $t=1$, we have for all eligible $i$ that $\sigma_i^t = \sigma_i^{(1)} = \gamma_i^{(1)} = \gamma_i^t$, and the descending procedure selects the maximum index.
	Assume the claim holds up to and including step $t-1$, when we advanced basket $i'$.
	Now let $i$ be the basket that is advanced in step $t$.
	Observe that, for any currently eligible basket $j \not\in \{i,i'\}$, $j$ was eligible in step $t-1$ by the nonincreasing property, and because $j$ was not advanced, we have $\gamma_j^t = \gamma_j^{t-1} = \sigma_j^{t-1} = \sigma_j^t$ as required.

	There are two cases: $i=i'$ and $i \neq i'$.
	In the first case, because the algorithm chose to advance $i$ on steps $t-1$ and $t$, we have $\gamma_i^{t-1} \geq \gamma_j^{t-1} = \gamma_j^t$ and $\sigma_i^t \geq \sigma_j^t = \gamma_j^t$ for all eligible $j \neq i$.
	This implies $\gamma_i^t = \min\{\gamma_i^{t-1}, \sigma_i^t\} \geq \gamma_j^t$.
	So $\gamma_i^t$ is maximum and we are done.
	In the second case, we have for all eligible $j \neq i$ that $\gamma_j^t \leq \sigma_j^t \leq \sigma_i^t = \gamma_i^t$, so $\gamma_i^t$ is maximum.
	Furthermore, in particular $\sigma_{i'}^t \leq \sigma_i^t \leq \gamma_{i'}^{t-1}$, so $\gamma_{i'}^t = \min\{\gamma_{i'}^{t-1}, \sigma_{i'}^t\} = \sigma_{i'}^t$.
	This completes the induction proof.
\end{proof}

We now use Lemma \ref{lemma:gamma-equality} to prove Lemma \ref{lemma:descending-nonexposed}.

\begin{proof}[Lemma \ref{lemma:descending-nonexposed}]
	To prove the lemma, suppose that $\I_i^{(\ell)} = 1$ and $\gamma_i^{(\ell)} < \kappa_i^{(\ell+1)}$.
	We must show that $\I_i^{(\ell+1)} = 1$.
	Let $t$ be the step on which the algorithm opened basket $i$'s box $\ell$.
	We will show that $i$ is eligible at time $t+1$ and that $\sigma_i^{t+1} > \sigma_j^{t+1}$ for all eligible $j \neq i$, implying that the descending procedure advances box $i$ on step $t+1$, so $\I_i^{(\ell+1)} = 1$ as required.

	First, because basket $i$ was advanced at step $t$, by definition of a descending procedure, $i$ is eligible at step $t+1$.
	Now consider any basket $j$ that is eligible at step $t+1$.
	Because eligible sets are nonincreasing in time, $j$ was eligible at time $t$.
	Because $j$ was not advanced at time $t$, Lemma \ref{lemma:gamma-equality} give us $\sigma_j^{t+1} = \gamma_j^{t+1} = \gamma_j^{t}$.
	And, because the descending procedure selected $i$ at time $t$, Lemma \ref{lemma:gamma-equality} also gives $\gamma_i^{t} \geq \gamma_j^{t}$.
	Finally, we supposed $\gamma_i^t < \kappa_i^{t+1} \leq \sigma_i^{t+1}$.
	This gives $\sigma_i^{t+1} > \sigma_j^{t+1}$ for all $j \neq i$, completing the proof.
\end{proof}

	\section{Omitted Proofs from Section \ref{sec:unordered}}\label{sec:unordered-appdx}

To show the approximation guarantee of Theorem \ref{theorem:rand-approx}, we first prove the upper-bound on the optimal welfare of Lemma \ref{lemma:opt-bound}.

\begin{proof}[Lemma \ref{lemma:opt-bound}]
	First, we decompose the welfare of the optimal algorithm into the contributions from $\O$ and $\hat{\O}$.
	Each edge that is inspected in the optimal algorithm is inspected in the order given by one of the two orientations.
	Given an edge $\{i,j\}$, let $W_{ij}^{\Opt}$ be the welfare contribution of the edge $\{i,j\}$ (i.e. total value claimed minus inspection cost paid) if box $(i,j)$ is inspected prior to $(j,i)$, else $W_{ij}^{\Opt} = 0$.
	Define $W_{ji}^{\Opt}$ analogously, so that the total welfare contribution of the edge is $W_{ij}^{\Opt} + W_{ji}^{\Opt}$.
	For each $(i,j) \in \O$, let $W_{ij}^{\O}$ (respectively $W_{ji}^{\hat{\O}}$) be the welfare contribution of basket $(i,j)$ (respectively, $(j,i)$) when running the optimal $\O$-oriented algorithm (respectively, $\hat{\O}$-oriented).
	Let $\Welf(\text{$\O$-desc.})$ be the expected welfare of the $\O$-Oriented Descending Procedure, and analogously for $\Welf(\text{$\hat{\O}$-desc.})$.
	\begin{align*}
		\Welf(\Opt)
		&=    \E \sum_{\{i,j\} \in E} \left(W_{ij}^{\Opt} + W_{ji}^{\Opt}\right) \\
		&=    \E \sum_{(i,j) \in \O} W_{ij}^{\Opt}  ~~+~~ \E\sum_{(j,i) \in \hat{\O}} W_{ji}^{\Opt}  \\
		&\leq \E \sum_{(i,j) \in \O} W_{ij}^{\O} ~~+~~ \E\sum_{(j,i) \in \hat{\O}} W_{ji}^{\hat{\O}}  \\
		&\leq 2 \Welf(\text{$\O$-desc.}) ~~+~ 2 \Welf(\text{$\hat{\O}$-desc.})    &  \text{Proposition \ref{prop:oriented-desc-half}}
	\end{align*}
\end{proof}

Finally, we can use this upper-bound to prove an approximation guarantee for the randomized algorithm.

\begin{proof}[Theorem \ref{theorem:rand-approx}]
	Since the orientation used by the randomized algorithm is drawn from the uniform distribution, the expected welfare is
	\begin{align*}
		\Welf(\Alg^\text{rand}) &= \E_{\O\sim \text{ unif}}\left[\Welf(\text{$\O$-desc.})\right]\\
		&= \E_{\O\sim \text{ unif}}\left[ \frac{1}{2}\Welf(\text{$\O$-desc.}) ~~+~  \frac{1}{2}\Welf(\text{$\hat{\O}$-desc.}) \right]\\
		&\geq \frac{1}{2}\E_{\O\sim \text{ unif}}\left[ \frac{1}{2} \Welf(\Opt)\right]&\text{Lemma \ref{lemma:opt-bound}}\\
		&= \frac{1}{4}\Welf(\Opt).
	\end{align*}
\end{proof}

\subsection{Best-of-Two-Worlds}

We also present a deterministic algorithm, which leverages similar tools to achieve a matching $1/4$-approximation to our randomized algorithm.

\begin{algdef}[Best-of-Two-Worlds Algorithm]\label{alg:best-of-two}
	Given an instance of Pandora's Matching Problem, pick an arbitrary orientation $\O$.
	Compute the expected welfare of the $\O$-Oriented Descending Procedure and the $\hat{\O}$-Oriented Descending Procedure.
	Run the one with larger expected welfare.
\end{algdef}

We leverage Lemma \ref{lemma:opt-bound} again to prove the approximation guarantee of Theorem \ref{theorem:best-of-two-approx}.

\begin{proof}[Theorem \ref{theorem:best-of-two-approx}]
	The welfare of the Best-of-Two-Worlds algorithm is
	\begin{align*}
		\max\left\{ \Welf(\text{$\O$-desc.}) ~,~ \Welf(\text{$\hat{\O}$-desc.}) \right\}  &\geq \frac{1}{2}\Welf(\text{$\O$-desc.}) ~~+~ \frac{1}{2}\Welf(\text{$\hat{\O}$-desc.})\\
		&\geq \frac{1}{4}\Welf(\Opt) &\text{Lemma \ref{lemma:opt-bound}}
	\end{align*}
\end{proof}

\subsection{No-Dessert Proofs} \label{subapp:no-dessert}

We detail the calculation of indices used in the proof of our ``no dessert'' theorem, Theorem \ref{thm:no-dessert}.

\begin{figure}
	\begin{center}
	{\footnotesize
		\parbox{0.25\linewidth}{
			\vspace{2em}
			observe $\sigma_{ij}^{(1)} = \tfrac{1}{\alpha}-1$
			\quad $\rightarrow$}
		\parbox{0.37\linewidth}{\underline{pay $c_{ij}=1$ to open first box} \\
		\\
		\\
		$\begin{cases}	\text{observe } v_{ij} = \tfrac{1}{\alpha^3},	& \text{w.prob. } \alpha  ~~ \rightarrow \\
						\sigma_{ij}^{(2)} = \tfrac{1}{\alpha^3}-\tfrac{1}{\alpha^2}+\tfrac{1}{\alpha}  \\
			\\
			\\
			\\
			\\
			\\
			\\
						\text{observe } v_{ij} = 0,						& \text{\hspace{-4ex}w.prob. } 1-\alpha ~~ \rightarrow \\
						\sigma_{ij}^{(2)} = \tfrac{-1}{\alpha^2}+\tfrac{1}{\alpha}
        \end{cases}$
		}
		\parbox{0.3\linewidth}{\underline{pay $c_{ji}=1-\alpha$ to open second box} \\
		\\
		$\begin{cases}	\text{observe } \bar{v} = \tfrac{1}{\alpha^3},	& \text{w.prob. } \alpha^2  \\
						\kappa_{ij}^{(2)} = \tfrac{1}{\alpha^3}-\tfrac{1}{\alpha^2}+\tfrac{1}{\alpha}, \\
						\kappa_{ij}^{(1)} = \tfrac{1}{\alpha}-1  \\
						\\
						\text{observe } \bar{v} = \tfrac{1}{\alpha},	& \text{w.prob. } 1-\alpha^2  \\
						\kappa_{ij}^{(2)} = \tfrac{1}{\alpha},  \\
						\kappa_{ij}^{(1)} = \tfrac{1}{\alpha}-1
		\end{cases}$ \\
		\\
		\\
		$\begin{cases}	\text{observe } \bar{v} = 0,										& \text{w.prob. } \alpha^2  \\
						\kappa_{ij}^{(2)} = \tfrac{-1}{\alpha^2}+\tfrac{1}{\alpha},  \\
						\kappa_{ij}^{(1)} = \tfrac{-1}{\alpha^2}+\tfrac{1}{\alpha}  \\
						\\
						\text{observe } \bar{v} = -\tfrac{1}{\alpha^3} + \tfrac{1}{\alpha}	& \text{w.prob. } 1-\alpha^2  \\
						\kappa_{ij}^{(2)} = \tfrac{-1}{\alpha^3}+\tfrac{1}{\alpha},  \\
						\kappa_{ij}^{(1)} = \tfrac{-1}{\alpha^3}+\tfrac{1}{\alpha}
		\end{cases}$
		}
	} %
	\end{center}
	\caption{The basket $(i,j)$. Here $\bar{v} = v_{ij} + v_{ji}$. After opening the first box, $\sigma_{ij}^{(2)}$ is determined. After opening the second box, $\bar{v}$ is determined and we can calculate $\kappa_{ij}^{(2)} = \min\{\bar{v}, \sigma_{ij}^{(2)}\}$ as well as $\kappa_{ij}^{(1)} = \min\{\bar{v}, \sigma_{ij}^{(2)}, \sigma_{ij}^{(1)}$.}
	\label{fig:ij-orient-calcs}
\end{figure}

\begin{lemma}
	For the Pandora basket $(i,j)$ resulting by orienting the edge so that $i$ is inspected prior to $j$, the initial Weitzman index and the capped values are the following:
	\begin{itemize}
		\item $\sigma_{ij}^{(1)} = \frac{1}{\alpha} - 1$.
		\item $\kappa_{ij}^{(1)} = \frac{1}{\alpha} - 1$ with probability $\alpha$, otherwise negative.
		\item $\E[(\kappa_{ij}^{(1)})^+] = 1 - \alpha$.
	\end{itemize}
\end{lemma}
\begin{proof}
	We utilize Figure \ref{fig:ij-orient-calcs}.
	We just need to prove that each $\sigma_{ij}^{(k)}$ and $\kappa_{ij}^{(k)}$ given in Figure \ref{fig:ij-orient-calcs} is correct.
	They are computed by backward induction: $\sigma_{ij}^{(2)}$, then $\kappa_{ij}^{(2)}$, then $\sigma_{ij}^{(1)}$, and finally $\kappa_{ij}^{(1)}$.
	First, we show that in the case $v_{ij} = \tfrac{1}{\alpha^3}$, it is correct that $\sigma_{ij}^{(2)} = \tfrac{1}{\alpha^3}-\tfrac{1}{\alpha^2}+\tfrac{1}{\alpha}$.
	Specializing Definition \ref{def:nested-index} to this setting:
	\begin{align*}
		\E \left[ (\bar{v} - \sigma_{ij}^{(2)})^+ ~\middle|~ v_{ij} = \tfrac{1}{\alpha^3} \right]
		&= \E \left[ (\tfrac{1}{\alpha^3} + v_{ji} - (\tfrac{1}{\alpha^3}-\tfrac{1}{\alpha^2}+\tfrac{1}{\alpha}))^+ \right]  \\
		&= \E \left[ (v_{ji} + \tfrac{1}{\alpha^2} - \tfrac{1}{\alpha})^+ \right]  \\
		&= \alpha^2 \left( \tfrac{1}{\alpha^2} - \tfrac{1}{\alpha} \right) + (1 - \alpha^2)(0)  \\
		&= 1 - \alpha  \\
		&= c_{ji},
	\end{align*}
	as required.
	For the case $v_{ij} = 0$, we claim $\sigma_{ij}^{(2)} = \tfrac{-1}{\alpha^2}+\tfrac{1}{\alpha}$:
	\begin{align*}
		\E \left[ (\bar{v} - \sigma_{ij}^{(2)})^+ ~\middle|~ v_{ij} = 0 \right]
		&= \E \left[ (v_{ji} - (\tfrac{-1}{\alpha^2}+\tfrac{1}{\alpha}))^+ \right]  \\
		&= \alpha^2 \left( 0 + \tfrac{1}{\alpha^2} - \tfrac{1}{\alpha} \right) + (1-\alpha^2)(0)  \\
		&= 1 - \alpha  \\
		&= c_{ji} .
	\end{align*}
	Now we can define $\kappa_{ij}^{(2)} = \min\{\bar{v}, \sigma_{ij}^{(2)}\}$ and check that Figure \ref{fig:ij-orient-calcs} is correct in all four states.
	Then, we just need to check that $\sigma_{ij}^{(1)} = \tfrac{1}{\alpha} - 1$ as claimed.
	With Definition \ref{def:nested-index}:
	\begin{align*}
		\E \left[ (\kappa_{ij}^{(2)} - \sigma_{ij}^{(1)})^+ \right]
		&= \E \left[ (\kappa_{ij}^{(2)} - \tfrac{1}{\alpha} + 1)^+ \right]  \\
		&= \alpha^3 \left(\tfrac{1}{\alpha^3} - \tfrac{1}{\alpha^2} + \tfrac{1}{\alpha} - \tfrac{1}{\alpha} + 1 \right)  \\
		&\quad\quad + \alpha(1-\alpha^2) \left(\tfrac{1}{\alpha} - \tfrac{1}{\alpha} + 1\right)  \\
		&\quad\quad + (1-\alpha)\alpha^2 (0) + (1-\alpha)(1-\alpha^2)(0)  \\
		&= \alpha^3 \left(\tfrac{1}{\alpha^3} - \tfrac{1}{\alpha^2} + 1 \right) + \alpha(1-\alpha^2) (1)  \\
		&= 1 - \alpha + \alpha^3 + \alpha - \alpha^3  \\
		&= 1  \\
		&= c_{ij} .
	\end{align*}
	We immediately calculate $\kappa_{ij}^{(1)} := \min\{\bar{v}, \sigma_{ij}^{(2)}, \sigma_{ij}^{(1)} \}$ in each state of the world, as given in Figure \ref{fig:ij-orient-calcs}.
	Observing that $\kappa_{ij}^{(1)} = \tfrac{1}{\alpha} - 1$ with probability $\alpha$ and is negative otherwise, we obtain $\E[ (\kappa_{ij}^{(1)})^+ ] = 1-\alpha$.
\end{proof}

\begin{figure}
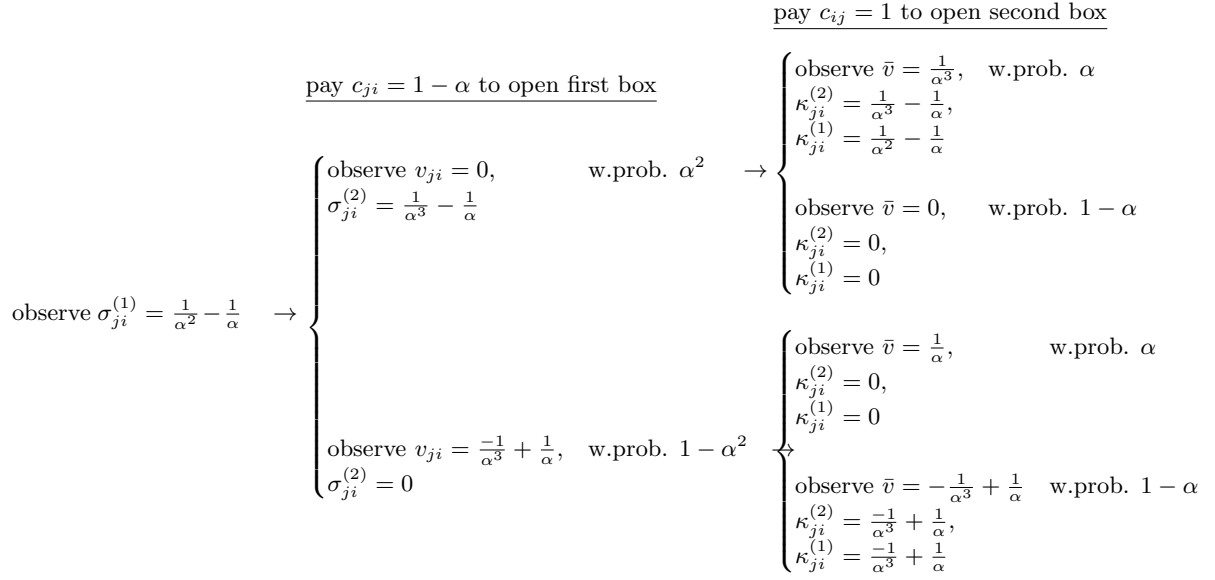

	\begin{center}
	{\footnotesize
		\parbox{0.25\linewidth}{
			\vspace{2em}
			observe $\sigma_{ji}^{(1)} = \tfrac{1}{\alpha^2}-\tfrac{1}{\alpha}$
			\quad $\rightarrow$}
		\parbox{0.37\linewidth}{\underline{pay $c_{ji}=1-\alpha$ to open first box} \\
		\\
		\\
		$\begin{cases}	\text{observe } \\ v_{ji} = 0,	& \text{w.prob. } \alpha^2  \hspace{3ex} \rightarrow \\
						\sigma_{ji}^{(2)} = \tfrac{1}{\alpha^3}-\tfrac{1}{\alpha}  \\
			\\
			\\
			\\
			\\
			\\
			\\
						\text{observe } \\ v_{ji} = \tfrac{-1}{\alpha^3} + \tfrac{1}{\alpha},	& \text{w.prob. } 1-\alpha^2 ~~ \rightarrow \\
						\sigma_{ji}^{(2)} = 0
        \end{cases}$
		}
		\parbox{0.3\linewidth}{\underline{pay $c_{ij}=1$ to open second box} \\
		\\
		$\begin{cases}	\text{observe } \bar{v} = \tfrac{1}{\alpha^3},	& \text{w.prob. } \alpha  \\
						\kappa_{ji}^{(2)} = \tfrac{1}{\alpha^3}-\tfrac{1}{\alpha}, \\
						\kappa_{ji}^{(1)} = \tfrac{1}{\alpha^2}-\tfrac{1}{\alpha}  \\
						\\
						\text{observe } \bar{v} = 0,					& \text{w.prob. } 1-\alpha  \\
						\kappa_{ji}^{(2)} = 0,  \\
						\kappa_{ji}^{(1)} = 0
		\end{cases}$ \\
		\\
		\\
		$\begin{cases}	\text{observe } \bar{v} = \tfrac{1}{\alpha},						& \text{w.prob. } \alpha  \\
						\kappa_{ji}^{(2)} = 0,  \\
						\kappa_{ji}^{(1)} = 0  \\
						\\
						\text{observe } \bar{v} = -\tfrac{1}{\alpha^3} + \tfrac{1}{\alpha}	& \text{w.prob. } 1-\alpha  \\
						\kappa_{ji}^{(2)} = \tfrac{-1}{\alpha^3}+\tfrac{1}{\alpha},  \\
						\kappa_{ji}^{(1)} = \tfrac{-1}{\alpha^3}+\tfrac{1}{\alpha}
		\end{cases}$
		}
	} %
	\end{center}
	\caption{The basket $(j,i)$. Again, $\bar{v} = v_{ij} + v_{ji}$.}
	\label{fig:ji-orient-calcs}
\end{figure}

\begin{lemma}
	For the Pandora basket $(j,i)$ resulting by orienting the edge so that $j$ is inspected prior to $i$, the initial Weitzman indices and expected capped values are the following:
	\begin{itemize}
		\item $\sigma_{ji}^{(1)} = \frac{1}{\alpha^2} - \frac{1}{\alpha}$.
		\item $\kappa_{ji}^{(1)} = \frac{1}{\alpha^2} - \frac{1}{\alpha}$ with probability $\alpha^3$, otherwise nonpositive.
		\item $\E[(\kappa_{ji}^{(1)})^+] = \alpha - \alpha^2$.
	\end{itemize}
\end{lemma}
\begin{proof}
	We first show that, conditioned on $v_{ji} = 0$, we have $\sigma_{ji}^{(2)} = \tfrac{1}{\alpha^3}-\tfrac{1}{\alpha}$.
	\begin{align*}
		\E \left[ (\bar{v} - \sigma_{ji}^{(2)})^+ ~\middle|~ v_{ji}=0 \right]
		&= \E \left[ (v_{ij} - \tfrac{1}{\alpha^3} + \tfrac{1}{\alpha})^+ \right]  \\
		&= \alpha \left( \tfrac{1}{\alpha^3} - \tfrac{1}{\alpha^3} + \tfrac{1}{\alpha} \right)  \\
		&= 1  \\
		&= c_{ij} .
	\end{align*}
	Next, conditioned on $v_{ji} = \tfrac{-1}{\alpha^3} + \tfrac{1}{\alpha}$, we have $\sigma_{ji}^{(2)} = 0$.
	\begin{align*}
		\E \left[ (\bar{v} - \sigma_{ji}^{(2)})^+ ~\middle|~ v_{ji} = \tfrac{-1}{\alpha^3} + \tfrac{1}{\alpha} \right]
		&= \E \left[ (v_{ij} - \tfrac{1}{\alpha^3} + \tfrac{1}{\alpha})^+ \right]  \\
		&= \alpha \left( \tfrac{1}{\alpha^3} - \tfrac{1}{\alpha^3} + \tfrac{1}{\alpha} \right)  \\
		&= 1  \\
		&= c_{ij} .
	\end{align*}
	Now, we check that Figure \ref{fig:ji-orient-calcs} correctly computes $\kappa_{ji}^{(2)} := \min\{\bar{v}, \sigma_{ji}^{(2)}\}$ in each state of the world.
	Then, to show $\sigma_{ji}^{(1)} = \tfrac{1}{\alpha^2} - \tfrac{1}{\alpha}$:
	\begin{align*}
		\E \left[ (\kappa_{ji}^{(2)} - \sigma_{ji}^{(1)})^+ \right]
		&= \E \left[ (\kappa_{ji}^{(2)} - \tfrac{1}{\alpha^2} + \tfrac{1}{\alpha})^+ \right]  \\
		&= \alpha^3 \left(\tfrac{1}{\alpha^3} - \tfrac{1}{\alpha} - \tfrac{1}{\alpha^2} + \tfrac{1}{\alpha} \right)  \\
		&= \alpha^3 \left(\tfrac{1}{\alpha^3} - \tfrac{1}{\alpha^2} \right)  \\
		&= 1 - \alpha  \\
		&= c_{ji} .
	\end{align*}
	We now check that Figure \ref{fig:ji-orient-calcs} correctly computes $\kappa_{ji}^{(1)} := \min\{\bar{v}, \sigma_{ji}^{(2)}, \sigma_{ji}^{(1)}\}$ in each state of the world.
	We obtain $\kappa_{ji}^{(1)} = \tfrac{1}{\alpha^2} - \tfrac{1}{\alpha}$ with probability $\alpha^3$, and nonpositive otherwise, so $\E[ (\kappa_{ji}^{(1)})^+ ] = \alpha - \alpha^2$.
\end{proof}

\subsection{No Dessert Alternate Proof} \label{subapp:no-dessert-alternate}

For a reader who would prefer to double-check without relying on the Nested Box machinery, we give an alternate proof of the ``No-Dessert'' Theorem: no deterministic edge-based fixed-orientation algorithm has nontrivial welfare.

\begin{proof}[Theorem \ref{thm:no-dessert}]
	The proof relies on a particular edge, shown in Figure \ref{fig:no-dessert}.
	Every edge-based fixed orientation algorithm either orients that edge in one direction or the other.
	We will show that depending on the outside option available, either choice can be a fatal mistake.

	\textbf{Instance 1.}
	We consider a graph consisting of just the edge in Figure \ref{fig:no-dessert}.

	There are three nontrivial algorithms to consider:
	\begin{enumerate}
		\item First inspect $(i,j)$, then if $v_{ij} > 0$, continue to inspecting $(j,i)$, then match the edge if $v_{ij} + v_{ji} > 0$.
			Denote its welfare by $\Welf(i,j)$.
		\item First inspect $(j,i)$, then if $v_{ji} = 0$, continue to $(i,j)$, then match the edge if $v_{ij} + v_{ji} > 0$.
			Denote its welfare by $\Welf(j,i)$.
		\item Inspect both endpoints, then match the edge if $v_{ij} + v_{ji} > 0$.
			Denote its welfare by $\Welf(\{i,j\})$.
			Observe this strategy is available to an algorithm that has oriented the edge in either direction.
	\end{enumerate}

	We can calculate the $\Welf(i,j)$ as follows.
	It matches the edge whenever $v_{ij} + v_{ji} > 0$.
	The possible positive values of $v_{ij} + v_{ji}$ are $\frac{1}{\alpha^3}$, which occurs with probability $\alpha^3$, and $\frac{1}{\alpha}$, which occurs with probability $\alpha(1-\alpha^2)$.
	The probability of paying $c_{ij}$ is $1$, and the probability of paying $c_{ji}$ is $\alpha$.
	Overall,
	\begin{align*}
		\Welf(i,j)
		&= \alpha^3 \left(\frac{1}{\alpha^3}\right) + \alpha(1-\alpha^2)\left(\frac{1}{\alpha}\right) - 1 - \alpha(1-\alpha)  \\
		&= 1 + 1 - \alpha^2 - 1 - \alpha + \alpha^2  \\
		&= 1 - \alpha .
	\end{align*}
	We calculate $\Welf(j,i)$ in the same way.
	The total value $\frac{1}{\alpha^3}$ is obtained with probability $\alpha^3$.
	In all other cases, the edge is not matched.
	The probability of paying $c_{ij}$ is $\alpha^2$ and the probability of paying $c_{ji}$ is $1$.
	\begin{align*}
		\Welf(j,i)
		&= \alpha^3 \left(\frac{1}{\alpha^3}\right) - \alpha^2 (1) - (1-\alpha)  \\
		&= 1 - \alpha^2 - 1 + \alpha  \\
		&= \alpha - \alpha^2 .
	\end{align*}
	Finally, for $\Welf(\{i,j\})$, whenever $v_{ij} + v_{ji} > 0$ the edge is matched, and both costs are always paid.
	\begin{align*}
		\Welf(\{i,j\})
		&= \alpha^3 \left(\frac{1}{\alpha^3}\right) + \alpha(1-\alpha^2)\left(\frac{1}{\alpha}\right) - 1 - (1-\alpha)  \\
		&= 1 + 1 - \alpha^2 - 1 - 1 + \alpha  \\
		&= \alpha - \alpha^2 .
	\end{align*}
	Because $\inf_{\alpha > 0} \frac{\Welf(j,i)}{\Welf(i,j)} = \inf_{\alpha > 0} \frac{\Welf(\{i,j\})}{\Welf(i,j)} = 0$, we conclude that \textbf{no edge-based fixed-orientation algorithm that orients $\{i,j\}$ as $(j,i)$ has a nonzero approximation guarantee on Instance 1.}

	\vskip1em   %
	\textbf{Instance 2.}
	We will consider a star graph with $i$ at the center.
	There is a special edge $\{i,k\}$ with $v_{ik} = \frac{1}{\alpha}$, $c_{ik} = 0$, $v_{ki} = 0$, $c_{ki} = 0$.
	In other words, $i$ has an ``outside option'' $\frac{1}{\alpha}$ it can match to at any time.
	Every other edge in the graph $\{i,j_1\},\dots,\{i,j_m\}$ is an independent copy of the edge $\{i,j\}$ in Figure \ref{fig:no-dessert}.
	The number of copies $m$ will be chosen later.

	\paragraph{Welfare of $(j,i)$ orientation}

	Note that one of the copies improves on the outside option if and only if $v_{ij} = \frac{1}{\alpha^3}$ and $v_{ji} = 0$ (a ``success'').
	To lower-bound the welfare achievable on this instance, we consider the algorithm that orients all copies as $(j,i)$, then iterates through until it finds a success.
	On each edge, it inspects $(j,i)$, then continues to $(i,j)$ if $v_{ji} = 0$, then stops and matches if $v_{ij} = \frac{1}{\alpha^3}$.
	If no edge is a success, the algorithm takes the outside option $\frac{1}{\alpha}$.

	Define the algorithm's initial welfare as $\frac{1}{\alpha}$, and for each edge, we can calculate the expected \emph{gain} from an edge as the expected increase in the welfare of the algorithm after processing this edge.
	The algorithm will increase its matched value by $\frac{1}{\alpha^3} - \frac{1}{\alpha}$ on success, will lose the cost $c_{ji} = 1-\alpha$ with probability $1$, and will lose the cost $c_{ij} = 1$ with probability $\alpha^2$.
	So for each copy, we have
	\begin{align*}
		\E [\text{gain}] &= \alpha^3 \left(\frac{1}{\alpha^3} - \frac{1}{\alpha}\right) - (1-\alpha) - \alpha^2(1)  \\
			&= 1 - \alpha^2 - 1 + \alpha - \alpha^2  \\
			&= \alpha - 2\alpha^2 .
	\end{align*}
	The algorithm obtains this expected gain every time it inspects a copy, and it inspects copies until it runs out or finds a success.
	If the number of copies is $m = \frac{1}{\alpha^{3-\epsilon}}$ for any $\epsilon > 0$, then the probability of finding a success is $o(1)$, which means the algorithm inspects all of the copies with high probability.
	So its total expected gain is $(1 - o(1)) m (\alpha - 2 \alpha^2) = \Omega\left(\frac{1}{\alpha^{2-\epsilon}}\right)$ for $\alpha \to 0$.

	\paragraph{Welfare of $(i,j)$ orientation}
	On the other hand, consider any edge-based fixed-orientation algorithm that orients $\{i,j\}$ as $(i,j)$.
	An analogous calculation shows that its expected gain is negative:\footnote{The intuition is that both endpoints must ``succeed'' for this edge to exceed the outside option, so either endpoint has a veto. It is cheaper to inspect $(j,i)$ first because it has a much higher chance of failure.}
	\begin{align*}
		\text{gain}(i,j) &= \alpha^3 \left(\frac{1}{\alpha^3} - \frac{1}{\alpha}\right) - 1 - \alpha(1-\alpha)  \\
			&= 1 - \alpha^2 - 1 - \alpha + \alpha^2  \\
			&= -\alpha .
	\end{align*}
	Therefore, the optimal such algorithm does no inspections and takes the outside option, yielding welfare $\frac{1}{\alpha}$.
	Because the welfare ratio tends to zero as $\alpha \to 0$, we conclude \textbf{no edge-based fixed-orientation algorithm that orients $\{i,j\}$ as $(i,j)$ has a nonzero approximation guarantee on Instance 2.}
\end{proof}

\end{document}